\documentclass[12pt]{article}
\usepackage{amsmath,amsfonts,amscd,theorem,color}
\usepackage[ddmmyyyy,hhmmss]{datetime}

\newcommand{\A}{\mathbb{A}}
\newcommand{\bA}{\boldsymbol{A}}
\newcommand{\bB}{\boldsymbol{B}}
\newcommand{\bC}{\boldsymbol{C}}
\newcommand{\be}{\boldsymbol{e}}
\newcommand{\bE}{\boldsymbol{E}}
\newcommand{\bF}{\boldsymbol{F}}
\newcommand{\bG}{\boldsymbol{G}}

\newcommand{\bone}{\boldsymbol{1}}
\newcommand{\bP}{\boldsymbol{P}}
\newcommand{\bphi}{{\boldsymbol{\phi}}}
\newcommand{\bPhi}{{\boldsymbol{\Phi}}}
\newcommand{\bPi}{\boldsymbol{\Pi}}

\newcommand{\bQ}{\boldsymbol{Q}}

\newcommand{\bS}{\boldsymbol{S}}
\newcommand{\bT}{\boldsymbol{T}}
\newcommand{\bV}{\boldsymbol{V}}
\newcommand{\bW}{\boldsymbol{W}}

\newcommand{\bX}{\boldsymbol{X}}

\newcommand{\bzD}{\boldsymbol{\Delta}}
\newcommand{\bzero}{\boldsymbol{0}}
\newcommand{\bzs}{\boldsymbol{\sigma}}
\newcommand{\C}{\mathbb{C}}
\newcommand{\CA}{\mathcal{A}}
\newcommand{\CB}{\mathcal{B}}
\newcommand{\CC}{\mathcal{C}}

\newcommand{\CF}{\mathcal{F}}

\newcommand{\CL}{\mathcal{L}}

\newcommand{\CN}{\mathcal{N}}

\newcommand{\CQ}{\mathcal{Q}}

\newcommand{\CW}{\mathcal{W}}

\newcommand{\F}{\mathbb{F}}
\newcommand{\hf}{\frac{1}{2}}

\newcommand{\hy}{\hat{y}}
\newcommand{\J}{\mathbb{J}}
\newcommand{\K}{\mathbb{K}}
\newcommand{\N}{\mathbb{N}}
\newcommand{\os}{{\overline{s}}}
\newcommand{\pt}{\partial}
\newcommand{\R}{\mathbb{R}}
\newcommand{\tPi}{\tilde{\boldsymbol{\Pi}}}

\newcommand{\tzs}{\tilde{\sigma}}

\newcommand{\V}{\mathbb{V}}

\newcommand{\za}{\alpha}
\newcommand{\zb}{\beta}

\newcommand{\zG}{\Gamma}

\newcommand{\zo}{\omega}
\newcommand{\zO}{\Omega}
\newcommand{\zs}{\sigma}
\newcommand{\zS}{\Sigma}
\newcommand{\zz}{\zeta}
\numberwithin{equation}{section}

\newtheorem{theorem}{\bf Theorem}
\numberwithin{theorem}{section}

\newtheorem{corollary}[theorem]{\bf Corollary}
\newtheorem{proposition}[theorem]{\bf Proposition}
{\theorembodyfont{\rmfamily} }
\newenvironment{proof}
{\begin{trivlist}\item[\hskip\labelsep\quad{\bf Proof.}
\hspace{0.5 em}]}{\hfill \rule{0.5em}{0.5em} \end{trivlist}}
\numberwithin{equation}{section}
\numberwithin{figure}{section}

\author{Francesco Demontis\footnote{Dipartimento di Matematica e Informatica,
Universit\`a di Ca\-glia\-ri, Via Ospedale 72, 09124 Cagliari, Italy. Email:
fdemontis@unica.it, cornelis110553@gmail.com}\,, Cornelis van der
Mee\footnotemark[1]}
\title{From the AKNS system to the matrix Schr\"o\-din\-ger equation with
vanishing potentials: Direct and inverse problems\footnote{LaTeX compilation
date and time: \today\ at \currenttime}}

\begin{document}
\date{}
\maketitle

\begin{abstract}
We relate the scattering theory of the focusing AKNS system with vanishing
boundary conditions to that of the matrix Schr\"o\-din\-ger equation. The
corresponding Miura transformation which allows this connection, converts the
focusing matrix nonlinear Schr\"o\-din\-ger (NLS) equation into a new nonlocal
integrable equation. We apply the matrix triplet method to derive the
multisoliton solutions of the nonlocal integrable equation, thus proposing a new
method to solve the matrix NLS equation.
\end{abstract}

\section{Introduction}\label{sec:1}

 For nearly 50 years nonlinear Schr\"o\-din\-ger (NLS) equations have served as
the basic models for describing surface waves on deep waters
\cite{Abl,ASg,ZS72}, signals along optical fibres \cite{HT,Has,Shaw}, particle
states in Bose-Einstein condensates \cite{PS1,PS2,KFCG}, plasma oscillations
\cite{Z}, and magnetic spin waves \cite{CTNP,ZP}. NLS equations with solutions
decaying at infinity have been studied in detail by means of the inverse
scattering transform \cite{AKNS,ASg,CD82,FT,APT}, where a time invariant
canonical transformation converts the initial-value problem of the matrix NLS
equation into the (elementary) time evolution of the scattering data of the
so-called AKNS system.

In this article we present an apparently new method for solving matrix NLS
equations inspired by the Miura transformation \cite{ASg,AC} which traditionally
allows one to derive Korteweg-de Vries (KdV) solutions from modified
Korteweg-de Vries (mKdV) solutions (but not necessarily vice versa). By
a similar transformation the direct and inverse scattering theory of the
focusing $n=m_1+m_2$ AKNS system can be related to the direct and inverse
scattering theory of a matrix Schr\"o\-din\-ger equation whose $n\times n$
matrix potential $\bQ$ satisfies the adjoint symmetry relation
\begin{equation}\label{1.1}
\bQ^\dagger=\bzs_3\bQ\bzs_3,
\end{equation}
where $\bzs_3=I_{m_1}\oplus(-I_{m_2})$ and the dagger denotes the matrix
conjugate transpose. The traditional application of the matrix Schr\"o\-din\-ger
equation to quantum graphs, quantum wires, and quantum mechanical scattering of
particles with internal structure \cite{Berk,BCFK,BkKm,BkLiu,BoKu,EKKST,Ger,
GerPav,GtSm,Hm02,Hm04,Hm05,KS99,KS00,Km04,Km05,KuNw05,KuNw10,KuSt} has led to
the almost exclusive development of matrix Schr\"o\-din\-ger scattering theory
for selfadjoint potentials for which $\bQ^\dagger=\bQ$. After the seminal
papers on the scalar Schr\"o\-din\-ger scattering theory by Faddeev \cite{Fa64},
Deift and Trubowitz \cite{DT79}, and others (see \cite{CS89}), in the matrix
case the half-line theory has been described in detail by Agranovich and
Marchenko \cite{AM} and by Aktosun and Weder \cite{AW18,AW20} and the full-line
theory by Wadati and Kamijo \cite{Wd}; the most general small energy asymptotics
of the scattering data, crucial to developing the direct and inverse scattering
theory rigorously, is due to Klaus \cite{K88} in the scalar case and to Aktosun
et al. \cite{AKV01} in the matrix case.

Energy losses in quantum graphs, quantum wires, and particles with internal
structure naturally lead to matrix Schr\"o\-din\-ger potentials whose imaginary
part $[\bQ-\bQ^\dagger]/2i$ has constant sign. In this article we thus require
a modified direct and inverse scattering theory when solving the matrix NLS
equation by using the matrix Schr\"o\-din\-ger equation for potentials
satisfying
\begin{equation}\label{1.2}
\bQ^\dagger=\bzs_3\bQ\bzs_3.
\end{equation}

In this article we depart from the direct and inverse scattering theory of the
focusing AKNS system \cite{AKNS,APT,FT}
\begin{equation}\label{1.3}
v_x=(-ik\bzs_3+\CQ)v,
\end{equation}
where $v=v(x,k)$ is a vector function with $n=m_1+m_2$ components and the
potential $\CQ$ anticommutes with $\bzs_3=I_{m_1}\oplus(-I_{m_2})$. Letting
$L=i\bzs_3[(d/dx)I_n-\CQ]$ stand for the AKNS Hamiltonian, we easily verify that
$L^2$ is the matrix Schr\"o\-din\-ger Hamiltonian given by
\begin{align*}
L^2v&=-\bzs_3[(d/dx)I_n-\CQ]\bzs_3[(d/dx)I_n-\CQ]v\\
&=-[(d/dx)I_n+\CQ][(d/dx)I_n-\CQ]v\\
&=-v_{xx}+\CQ^2v-\CQ v_x+(\CQ v)_x=-v_{xx}+\bQ v,
\end{align*}
where
\begin{equation}\label{1.4}
\bQ=\CQ^2+\CQ_x
\end{equation}
is a matrix Schr\"o\-din\-ger potential obtained from $\CQ$ by the Miura
transform \eqref{1.4}. Assuming that the entries of $\CQ$ and $\CQ_x$ belong to
the weighted $L^1$-space $L^1(\R;(1+|x|)dx)$, we arrive from \eqref{1.3} at the
matrix Schr\"o\-din\-ger equation
\begin{equation}\label{1.5}
-v_{xx}+\bQ v=k^2v,
\end{equation}
where $\bQ$ is a so-called Faddeev class potential (i.e., its entries belong to
$L^1(\R;(1+|x|)dx)$) satisfying the adjoint symmetry relation \eqref{1.2}.

It is well-known \cite{ZS72,AKNS,APT,FT} that for a suitable time evolution of
the AKNS scattering data the potential $\CQ=\CQ(x;t)$ satisfies the matrix NLS
equation
\begin{equation}\label{1.6}
i\CQ_t+\CQ_{xx}-2\CQ^3=0_{n\times n},
\end{equation}
where $\CQ$ anticommutes with $\bzs_3$. Direct AKNS scattering theory involves
Jost solutions and certain scattering coefficients which are analytic in $k$ in
either upper half complex plane $\C^+$ or the lower half complex plane $\C^-$,
with continuous extensions up to the real $k$-line \cite{AKNS,APT,FT}. On the
other hand, traditionally direct matrix Schr\"o\-din\-ger scattering theory
involves Jost solutions and certain scattering coefficients which are continuous
in $k\in\C^+\cup\R$ and analytic in $k\in\C^+$ (\cite{AM,AW18,AW20} for the
half-line theory and \cite{Wd,AKV01} for the full line theory). In fact, the
values $\pm k$ of the AKNS spectral variable can be made to correspond to the
matrix Schr\"o\-din\-ger spectral variable $k$, which greatly simplifies direct
and inverse scattering theory. It is therefore natural to wonder which nonlinear
evolution equation $\bQ(x;t)$ is satisfied when $\CQ(x;t)$ is a matrix NLS
solution. We are thus lead to a nonlocal integrable equation for $\bQ(x;t)$ to
be derived from a Lax pair.

The principal objective of this article is to derive the multisoliton solutions
of the focusing NLS system by applying the Miura transform \eqref{1.4} to the
matrix Schr\"o\-din\-ger equation which is then solved explicitly by the
so-called matrix triplet method (see \cite{SIMAI} and references therein). These
solutions are obtained by solving the corresponding Marchenko integral
equations. Once obtained, these solutions are verified by direct substitution
into the nonlocal integrable equation satisfied by $\bQ(x;t)$.

Let us discuss the contents of this article. In Sec.~\ref{sec:2} we derive the
nonlocal integrable equation for $\bQ$ from a Lax pair. Next, in
Sections~\ref{sec:3}-\ref{sec:4} we develop the direct and inverse scattering
theory of the matrix Schr\"o\-din\-ger equation \eqref{1.5} for Faddeev class
potentials $\bQ$ satisfying \eqref{1.2}. In particular, we introduce the Jost
solutions and the scattering coefficients, write them as Fourier transforms of
$L^1$-functions, and derive the Marchenko integral equations to solve the
inverse scattering problem. The rather technical proof of the adjoint symmetry
relations of the Marchenko integral kernels is deferred to Appendix~\ref{sec:B}.
The relationship between our Jost solutions and scattering coefficients and
those prevailing in AKNS scattering theory will be worked out in
Sec.~\ref{sec:5}. We then go on to derive the time evolution of the scattering
data in Sec.~\ref{sec:6}. In Sec.~\ref{sec:7} we apply the so-called matrix
triplet method to derive the multisoliton solutions. Three appendices have been
attached. In Appendix~\ref{sec:A} we derive the nonlocal integrable equation for
$\bQ$ directly from the NLS equation for $\CQ$ without using Lax pairs. In
Appendix~\ref{sec:B} we derive the adjoint symmetry relations for the Marchenko
integral kernels. Finally, in Appendix~\ref{sec:C} we prove by direct
substitution that the multisoliton solutions satisfy the nonlocal integrable
equation. Moreover, the emphasis on the methods used in this article allows one
to construct the reflectionless solutions of the hierarchy of equations with
Lax pairs $\{(L^j,A)\}_{j=1}^\infty$, where $L$ is the AKNS Hamiltonian and $A$
describes the time evolution (see Sec.~\ref{sec:2}). In fact, we present the
details only for the first two equations of the hierarchy.

We adopt blackboard boldface symbols for many of the quantities pertaining to
the AKNS system, thus deviating from the notational system adopted in
\cite{APT}. Boldface symbols are reserved for many of the quantities pertaining
to the matrix Schr\"o\-din\-ger equation. We also deviate from the praxis of
\cite{AKNS,APT} in allowing right and left to correspond to the real line
endpoints of transmission, both in the AKNS case and in the (matrix)
Schr\"o\-din\-ger case. Thus we deviate from the right vs. left conventions
adopted in \cite{AKNS,APT}.

\section{Lax Pair for the new integrable model}\label{sec:2}

As to be indicated shortly, the matrix NLS system is governed by a Lax pair
$\{L,A\}$ of linear operators \cite{Lx,ASg,EH}, where $Lv=kv$ is the AKNS
eigenvalue problem and $v_t=Av$ describes the time evolution. Then the matrix
NLS system is equivalent to the zero curvature condition
$$L_t+LA-AL=\bzero,$$
where $\bzero$ denotes the zero operator on a suitable function space. We now
observe that $\{L^2,A\}$ is also a Lax pair of linear operators, where
$L^2v=k^2v$ is the matrix Schr\"o\-din\-ger eigenvalue problem and $v_t=Av$
describes the time evolution. Then the accompanying nonlinear evolution equation
is equivalent to the zero curvature condition
$$(L^2)_t+L^2A-AL^2=\bzero.$$

Putting
\begin{subequations}\label{2.1}
\begin{align}
L&=i\bzs_3(\pt_xI_n-\CQ),\label{2.1a}\\
A&=i\bzs_3\left(2\pt_x^2I_n-2\CQ\pt_x-\CB_0\right),\label{2.1b}
\end{align}
\end{subequations}
where $n=m_1+m_2$ and $\CB_0$ is to be determined, we compute
\begin{align*}
&\bzero=-(L_t+LA-AL)\\
&=i\bzs_3\CQ_t
+\bzs_3(\pt_xI_n-\CQ)\bzs_3\left(2\pt_x^2I_n-2\CQ\pt_x-\CB_0\right)\\
&-\bzs_3\left(2\pt_x^2I_n-2\CQ\pt_x-\CB_0\right)\bzs_3(\pt_xI_n-\CQ)\\
&=i\bzs_3\CQ_t
+(\pt_xI_n+\CQ)\left(2\pt_x^2I_n-2\CQ\pt_x-\CB_0\right)\\
&-\left(2\pt_x^2I_n+2\CQ\pt_x-\bzs_3\CB_0\bzs_3\right)(\pt_xI_n-\CQ)\\
&=i\bzs_3\CQ_t+2\pt_x^3I_n-2\CQ_x\pt_x-2\CQ\pt_x^2-(\CB_0)_x
-\CB_0\pt_x+2\CQ\pt_x^2-2\CQ^2\pt_x-\CQ\CB_0
-2\pt_x^3I_n\\&+2\CQ_{xx}+4\CQ_x\pt_x+2\CQ\pt_x^2-2\CQ\pt_x^2
+2\CQ\CQ_x+2\CQ^2\pt_x+\bzs_3\CB_0\bzs_3\pt_x-\bzs_3\CB_0\bzs_3\CQ\\
&=i\bzs_3\CQ_t-2(-\CQ_x+\tfrac{1}{2}[\CB_0-\bzs_3\CB_0\bzs_3])\pt_x\\
&-(\CB_0)_x-\CQ\CB_0+2\CQ_{xx}+2\CQ\CQ_x-\bzs_3\CB_0\bzs_3\CQ.
\end{align*}
To cancel the $\pt_x$ terms, $\CB_0$ should be the sum of $\CQ_x$ and a matrix
commuting with $\bzs_3$. Thus taking $\CB_0=\CQ^2+\CQ_x=\bQ$, we get the matrix
NLS equation
\begin{equation}\label{2.2}
i\bzs_3\CQ_t+\CQ_{xx}-2\CQ^3=0_{n\times n}.
\end{equation}

Putting $\CL=L^2=-\pt_x^2+\bQ$, we compute
\begin{align*}
i\bzs_3&[\CL_t+\CL A-A\CL]=i\bzs_3\bQ_t\\
&-(-\pt_x^2+\bzs_3\bQ\bzs_3)\left[2\pt_x^2-2\CQ\pt_x-\bQ\right]
+\left[2\pt_x^2-2\CQ\pt_x-\bQ\right](-\pt_x^2+\bQ)\\
&=i\bzs_3\bQ_t+4(-\CQ_x+\tfrac{1}{2}[\bQ-\bzs_3\bQ\bzs_3])\pt_x^2\\
&+2(-\CQ_{xx}+\bQ_x+\bzs_3\bQ\bzs_3\CQ-\CQ\bQ)\pt_x\\
&+\bQ_{xx}+\bzs_3\bQ\bzs_3\bQ-2\CQ\bQ_x-\bQ^2.
\end{align*}
Then \eqref{1.4} and $\bzs_3\CQ=-\CQ\bzs_3$ imply that the terms in the third
member involving $\pt_x^2$ and $\pt_x$ vanish. Therefore, we have arrived at the
nonlinear evolution equation
\begin{equation}\label{2.3}
i\bzs_3\bQ_t+\bQ_{xx}-\bQ^2+\bzs_3\bQ\bzs_3\bQ-2\CQ\bQ_x=0_{n\times n},
\end{equation}
where
\begin{equation}\label{2.4}
\CQ(x;t)=-\int_x^\infty dy\,\hf\left(\bQ-\bzs_3\bQ\bzs_3\right)
=\int_{-\infty}^x dy\,\hf\left(\bQ-\bzs_3\bQ\bzs_3\right).
\end{equation}
A proof of the integrability of \eqref{2.3} based directly on \eqref{1.4} and
\eqref{1.6} will be given in Appendix~\ref{sec:A}. Finally, using \eqref{1.4}
in \eqref{2.3} and singling out the block off-diagonal component we easily get
the derivative of \eqref{2.2} with respect to $x$ and hence \eqref{2.2} itself;
the block diagonal component yields the commutator of \eqref{2.2} and $\CQ$.

\section{Direct Scattering Problem}\label{sec:3}

In this section we introduce the Jost solutions and scattering coefficients for
the matrix Schr\"o\-din\-ger equation \eqref{1.5} with Faddeev class potential
$\bQ$ satisfying \eqref{1.2}. For the scalar Schr\"o\-din\-ger equation with
real Faddeev class potential the direct scattering theory is well documented
\cite{Fa64,DT79,CD82,CS89,AC}. For selfadjoint potentials, the matrix theory on
the half-line is discussed at length in \cite{AM,AW18,AW20}, whereas that on the
full-line can be found in \cite{Wd,AKV01}. Here we give the modifications caused
by the adjoint symmetry \eqref{1.2}.

\subsection{Jost solutions of the matrix Schr\"odinger equation}\label{sec:3.1}

Let us define the {\it Jost solution from the left} $F_l(x,k)$ and the {\it Jost
solution from the right} $F_r(x,k)$ as those solutions of the matrix
Schr\"o\-din\-ger equation \eqref{1.5} which satisfy the asymptotic conditions
\begin{subequations}\label{3.1}
\begin{alignat}{3}
 F_l(x,k)&=e^{ikx}\left[I_n+o(1)\right],&\qquad&x\to+\infty,\label{3.1a}\\
 F_r(x,k)&=e^{-ikx}\left[I_n+o(1)\right],&\qquad&x\to-\infty,\label{3.1b}
\end{alignat}
\end{subequations}
where $n=m_1+m_2$. Calling $m_l(x,k)=e^{-ikx}F_l(x,k)$ and
$m_r(x,k)=e^{ikx}F_r(x,k)$ {\it Faddeev functions}, we easily define them as
the unique solutions of the Volterra integral equations
\begin{subequations}\label{3.2}
\begin{align}
m_l(x,k)&=I_n+\int_x^\infty dy\,\frac{e^{2ik(y-x)}-1}{2ik}\bQ(y)m_l(y,k),
\label{3.2a}\\
m_r(x,k)&=I_n+\int_{-\infty}^x dy\,\frac{e^{2ik(x-y)}-1}{2ik}\bQ(y)m_r(y,k).
\label{3.2b}
\end{align}
\end{subequations}
Then, for each $x\in\R$, $m_l(x,k)$ and $m_r(x,k)$ are continuous in
$k\in\C^+\cup\R$, are analytic in $k\in\C^+$, and tend to $I_n$ as $k\to\infty$
from within $\C^+\cup\R$. For $0\neq k\in\R$ we can reshuffle \eqref{3.2} and
arrive at the asymptotic relations
\begin{subequations}\label{3.3}
\begin{align}
 F_l(x,k)&=e^{ikx}A_l(k)+e^{-ikx}B_l(k)+o(1),&\qquad&x\to-\infty,\label{3.3a}\\
 F_r(x,k)&=e^{-ikx}A_r(k)+e^{ikx}B_r(k)+o(1),&\qquad&x\to+\infty,\label{3.3b}
\end{align}
\end{subequations}
where
\begin{subequations}\label{3.4}
\begin{align}
A_{r,l}(k)&=I_n-\frac{1}{2ik}\int_{-\infty}^\infty dy\,\bQ(y)m_{r,l}(y,k),
\label{3.4a}\\
B_{r,l}(k)&=\frac{1}{2ik}\int_{-\infty}^\infty dy\,
e^{\mp2iky}\bQ(y)m_{r,l}(y,k).\label{3.4b}
\end{align}
\end{subequations}
Thus $A_{r,l}(k)$ is continuous in $0\neq k\in\C^+\cup\R$, is analytic in
$k\in\C^+$, and tends to $I_n$ as $k\to\infty$ from within $\C^+\cup\R$, while
$2ik[I_n-A_{r,l}(k)]$ has the finite limit
$-\bzD_{r,l}=\int_{-\infty}^\infty dy\,\bQ(y)m_{r,l}(y,k)$ as $k\to0$ from
within $\C^+\cup\R$. By the same token, $B_{r,l}(k)$ is continuous in
$0\neq k\in\R$, vanishes as $k\to\pm\infty$, and satisfies
$2ikB_{r,l}(k)\to-\bzD_{r,l}$ as $k\to0$ along the real $k$-axis.

Putting
$$\bF_l(x,k)=\begin{pmatrix}F_l(x,-k)&F_l(x,k)\\F_l^\prime(x,-k)&F_l^\prime(x,k)
\end{pmatrix},\ \bF_r(x,k)=\begin{pmatrix}F_r(x,k)&F_r(x,-k)\\F_r^\prime(x,k)
&F_r^\prime(x,-k)\end{pmatrix},$$
where the prime denotes differentiation with respect to $x$, we obtain
\begin{subequations}\label{3.5}
\begin{align}
\bF_r(x,k)&=\bF_l(x,k)\begin{pmatrix}A_r(k)&B_r(-k)\\B_r(k)&A_r(-k)
\end{pmatrix},\label{3.5a}\\
\bF_l(x,k)&=\bF_r(x,k)\begin{pmatrix}A_l(-k)&B_l(k)\\B_l(-k)&A_l(k)
\end{pmatrix},\label{3.5b}
\end{align}
\end{subequations}
where $0\neq k\in\R$. Using that $\bF_{r,l}(x,k)$ satisfies the linear first
order system
\begin{equation}\label{3.6}
\begin{pmatrix}V\\V^\prime\end{pmatrix}^\prime=\begin{pmatrix}0_{n\times n}
&I_n\\ \bQ(x)-k^2I_n&0_{n\times n}\end{pmatrix}
\begin{pmatrix}V\\V^\prime\end{pmatrix}
\end{equation}
with traceless system matrix, we see that, for $0\neq k\in\R$, $\bF_{r,l}(x,k)$
has a determinant not depending on $x\in\R$. Using \eqref{3.1} we easily verify
that $\det\bF_{r,l}(x,k)=(2ik)^n$ for $0\neq k\in\R$.

Putting
\begin{align}
\bPhi(x,k)&=\begin{pmatrix}F_r(x,k)&F_l(x,k)\\F_r^\prime(x,k)&F_l^\prime(x,k)
\end{pmatrix}\nonumber\\&=\bF_r(x,k)\begin{pmatrix}I_n&B_l(k)\\0_{n\times n}
&A_l(k)\end{pmatrix}=\bF_l(x,k)\begin{pmatrix}A_r(k)&0_{n\times n}\\B_r(k)&I_n
\end{pmatrix},\label{3.7}
\end{align}
we easily see that, for $0\neq k\in\R$, $\det A_r(k)=\det A_l(k)$. By analytic
continuation we get $\det A_r(k)=\det A_l(k)$ for $0\neq k\in\C^+\cup\R$.

Letting $\zs_2=\left(\begin{smallmatrix}0&-i\\i&0\end{smallmatrix}\right)$
denote the second Pauli matrix, we define
$$\zs_2\otimes\bzs_3=\begin{pmatrix}0_{n\times n}&-i\bzs_3\\
i\bzs_3&0_{n\times n}\end{pmatrix}$$
as the Kronecker product of $\zs_2$ and $\bzs_3$ (cf. \cite{HJ}).

\begin{proposition}\label{th:3.1}
 For $0\neq k\in\C^+\cup\R$, let $\bV(x,k)$ and $\bW(x,k)$ be two size
compatible matrix solutions of the linear first order system {\rm\eqref{3.6}}.
Then
$$\bW(x,-k^*)^\dagger(\zs_2\otimes\bzs_3)\bV(x,k)$$
is independent of $x\in\R$. Similarly, for $0\neq k\in\R$ the matrix
$$\bW(x,k)^\dagger(\zs_2\otimes\bzs_3)\bV(x,k)$$
is independent of $x\in\R$.
\end{proposition}

\begin{proof}
It is easily verified that
$$(\zs_2\otimes\bzs_3)\begin{pmatrix}0_{n\times n}&I_n\\
\bQ(x)-k^2I_n&0_{n\times n}\end{pmatrix}(\zs_2\otimes\bzs_3)
=-\begin{pmatrix}0_{n\times n}&I_n\\ \bQ(x)-(-k^*)^2I_n&0_{n\times n}
\end{pmatrix}^\dagger.$$
Then
\begin{align*}
\frac{\pt}{\pt x}&\left[\bW(x,-k^*)^\dagger(\zs_2\otimes\bzs_3)\bV(x,k)\right]\\
&=\bW(x,-k^*)^\dagger\begin{pmatrix}0_{n\times n}&I_n\\ \bQ(x)-(-k^*)^2I_n
&0_{n\times n}\end{pmatrix}^\dagger(\zs_2\otimes\bzs_3)\bV(x,k)\\
&+\bW(x,-k^*)^\dagger(\zs_2\otimes\bzs_3)\begin{pmatrix}0_{n\times n}&I_n\\
\bQ(x)-k^2I_n&0_{n\times n}\end{pmatrix}\bV(x,k)=0_{2n\times2n},
\end{align*}
as claimed. The second part follows in the same way.
\end{proof}

Let us now apply either part of Proposition \ref{th:3.1} to derive identities
for the $A$ and $B$ coefficients by equating the asymptotics as $x\to+\infty$ to
the asymptotics as $x\to-\infty$. Using the second part for $\bV=\bW=\bPhi$
we get
\begin{subequations}\label{3.8}
\begin{align}
A_{r,l}(k)^\dagger\bzs_3A_{r,l}(k)&-B_{r,l}(k)^\dagger\bzs_3B_{r,l}(k)=\bzs_3,
\label{3.8a}\\
B_{r,l}(k)^\dagger&=-\bzs_3B_{l,r}(k)\bzs_3,\label{3.8b}
\end{align}
\end{subequations}
where $0\neq k\in\R$. Using the second part for $\bV=\bW=\bF_{r,l}$, we get
\begin{equation}\label{3.9}
A_{r,l}(k)^\dagger\bzs_3B_{r,l}(-k)=B_{r,l}(k)^\dagger\bzs_3A_{r,l}(-k),
\end{equation}
where $0\neq k\in\R$. Using the second part for $\bV=\bF_l$ and $\bW=\bF_r$
we obtain
\begin{equation}\label{3.10}
A_r(k)^\dagger=\bzs_3A_l(-k)\bzs_3,\qquad B_r(k)^\dagger=-\bzs_3B_l(k)\bzs_3,
\end{equation}
where $0\neq k\in\R$. Finally, using the first part for $\bV=\bW=\bPhi$ we get
\begin{equation}\label{3.11}
A_{r,l}(-k^*)^\dagger=\bzs_3A_{l,r}(k)\bzs_3,
\end{equation}
where $0\neq k\in\C^+\cup\R$.

Introducing the {\it reflection coefficients}
\begin{equation}\label{3.12}
R_{r,l}(k)=B_{r,l}(k)A_{r,l}(k)^{-1}=-A_{l,r}(k)^{-1}B_{l,r}(-k)
\end{equation}
and the {\it transmission coefficients} $A_{r,l}(k)^{-1}$, we obtain the
{\it Riemann-Hilbert problem}
\begin{equation}\label{3.13}
\begin{pmatrix}F_l(x,-k)&F_r(x,-k)\end{pmatrix}=\begin{pmatrix}F_r(x,k)&F_l(x,k)
\end{pmatrix}\begin{pmatrix}A_r(k)^{-1}&-R_l(k)\\-R_r(k)&A_l(k)^{-1}
\end{pmatrix},
\end{equation}
where the matrix $\bS(k)$ containing the $A$ and $R$ quantities is called the
{\it scattering matrix} and the discussion of the nonsingularity of $A_{r,l}(k)$
will be presented shortly. Then it is easily verified that
\begin{equation}\label{3.14}
R_{r,l}(k)^\dagger=\bzs_3R_{r,l}(-k)\bzs_3,
\end{equation}
and
\begin{equation}\label{3.15}
\bS(k)^\dagger(\bzs_3\oplus\bzs_3)\bS(k)=\bzs_3\oplus\bzs_3,
\end{equation}
provided $0\neq k\in\R$ and $\det A_{r,l}(k)\neq0$.

Above we have defined $\bzD_{r,l}$ as follows:
$$\bzD_{r,l}=\lim_{k\to0}\,2ikA_{r,l}(k)=-\lim_{k\to0^\pm}\,2ikB_{r,l}(k),$$
where the first limit may be taken from the upper half-plane. Then the matrices
$\bzD_{r,l}$ have the same determinant. If $\bzD_{r,l}$ is nonsingular, we say
that we are in the {\it generic case}; if instead $\bzD_{r,l}$ is singular,
we are in the {\it exceptional case}. We are said to be in the
{\it superexceptional case} if $\bzD_{r,l}=0_{n\times n}$ and $A_{r,l}(k)$ tends
to a nonsingular matrix, $A_{r,l}(0)$ say, as $k\to0$ from within $\C^+\cup\R$.

Throughout this article, we assume the absence of {\it spectral singularities},
i.e., the absence of nonzero real $k$ for which $\det A_{r,l}(k)=0$. Under this
condition the reflection coefficients $R_{r,l}(k)$ are continuous in
$0\neq k\in\R$.

\subsection{Triangular representations and Wiener algebras}\label{sec:3.2}

The Jost solutions allow the triangular representations
\begin{subequations}\label{3.16}
\begin{align}
 F_l(x,k)&=e^{ikx}I_n+\int_x^\infty dy\,e^{iky}K(x,y),\label{3.16a}\\
 F_r(x,k)&=e^{-ikx}I_n+\int_{-\infty}^x dy\,e^{-iky}J(x,y),\label{3.16b}
\end{align}
\end{subequations}
where for every $x\in\R$
\begin{equation}\label{3.17}
\int_x^\infty dy\,\|K(x,y)\|+\int_{-\infty}^x dy\,\|J(x,y)\|<+\infty.
\end{equation}
To derive \eqref{3.16}, we introduce the auxiliary matrix functions $K(x,y)$ and
$J(x,y)$ as the unique solutions of the Volterra integral equations
\begin{subequations}\label{3.18}
\begin{align}
K&(x,y)=\hf\int_{\hf[x+y]}^\infty dz\,\bQ(z)+\hf\int_x^\infty dz\,\bQ(z)
\int_{\max(z,x+y-z)}^{z+y-x}d\hy\,K(z,\hy)\nonumber\\
&=\hf\int_{\hf[x+y]}^\infty dz\,\bQ(z)+\hf\int_x^y dz\int_{x+\hf(y-z)}^\infty
dw\,\bQ(w)K(w,w+z-x),\label{3.18a}\\
J&(x,y)=\hf\int_{-\infty}^{\hf[x+y]}dz\,\bQ(z)+\hf\int_{-\infty}^x dz\,\bQ(z)
\int_{z+y-x}^{\min(z,x+y-z)}d\hy\,J(z,\hy)\nonumber\\
&=\hf\int_{-\infty}^{\hf[x+y]}dz\,\bQ(z)+\hf\int_y^x dz
\int_{-\infty}^{x-\hf[z-y]}dw\,\bQ(w)J(w,w+z-x).\label{3.18b}
\end{align}
\end{subequations}
Solving \eqref{3.18} by iteration we easily derive \eqref{3.16} and
\eqref{3.17}. Moreover,
\begin{equation}\label{3.19}
K(x,x)=\hf\int_x^\infty dy\,\bQ(y),\qquad J(x,x)=\hf\int_{-\infty}^x dy\,\bQ(y).
\end{equation}

Let us now introduce some necessary terminology and well-known results on
 Fourier transforms of $L^1$-functions. By the (continuous) Wiener algebra $\CW$
we mean the complex vector space of constants plus Fourier transforms of
$L^1$-functions
$$\CW=\{c+\hat{h}:c\in\C,\ h\in L^1(\R)\}$$
endowed with the norm $|c|+\|h\|_1$. Here we define the Fourier transform as
follows: $(\CF h)(k)=\hat{h}(k)=\int_{-\infty}^\infty dy\,e^{ik y}h(y)$.
The invertible elements of the commutative Banach algebra $\CW$ with unit
element are exactly those $c+\hat{h}\in\CW$ for which $c\neq0$ and
$c+\hat{h}(k)\neq0$ for each $k\in\R$ \cite{GRS}.

The algebra $\CW$ has the two closed subalgebras $\CW^+$ and $\CW^-$ consisting
of those $c+\hat{h}\in\CW$ for which $h$ is supported on $\R^+$ and $\R^-$,
respectively. The invertible elements of $\CW^\pm$ are exactly those
$c+\hat{h}\in\CW^\pm$ for which $c\neq0$ and $c+\hat{h}(k)\neq0$ for each
$k\in\C^\pm\cup\R$ \cite{GRS}. Letting $\CW^\pm_0$ and $\CW_0$ stand for the
(nonunital) closed subalgebras of $\CW^\pm$ and $\CW$ consisting of those
$c+\hat{h}$ for which $c=0$, we obtain the direct sum decompositions
$$\CW=\C\oplus\CW^+_0\oplus\CW^-_0,\qquad\CW_0=\CW^+_0\oplus\CW^-_0.$$

By $\Pi_\pm$ we now denote the (bounded) projections of $\CW$ onto $\CW^\pm_0$
along $\C\oplus\CW^\mp_0$. Then $\Pi_+$ and $\Pi_-$ are complementary
projections. In fact,
\begin{equation}\label{3.20}
(\Pi_\pm f)(k)=\frac{1}{2\pi i}\int_{-\infty}^\infty d\zz\,
\frac{f(\zz)}{\zz-(k\pm i0^+)},
\end{equation}
where $f\in\CW_0\cap L^p(\R)$ for some $p\in(1,+\infty)$. These direct sum
decompositions can be coupled by the Fourier transform to the natural direct
sum decomposition of $L^1(\R)$ as follows:
$$\begin{array}{ccccc}
L^1(\R)&=&L^1(\R^-)&\oplus&L^1(\R^+)\\
\Big\downarrow\vcenter{\rlap{$\CF$}}&&\Big\downarrow\vcenter{\rlap{$\CF$}}&&
\Big\downarrow\vcenter{\rlap{$\CF$}}\\
\CW_0&=&\CW^-_0&\oplus&\CW^+_0
\end{array}$$

Throughout this article we denote the vector spaces of $n\times m$ matrices with
entries in $\CW$, $\CW^\pm$, and $\CW^\pm_0$ by $\CW^{n\times m}$,
${\CW^\pm}^{n\times m}$, and ${\CW^\pm_0}^{n\times m}$, respectively. We write
$L^1(\R)^{n\times m}$ and $L^1(\R^\pm)^{n\times m}$ for the vector spaces of
$n\times m$ matrices with entries in $L^1(\R)$ and $L^1(\R^\pm)$, respectively.
Using a submultiplicative matrix norm, we can turn all of these vector spaces
into Banach spaces. It is then clear that $\CW^{n\times n}$ and
${\CW^\pm}^{n\times n}$ are noncommutative Banach algebras with unit element and
${\CW^\pm_0}^{n\times n}$ are (nonunital) noncommutative Banach algebras. The
projections $\Pi^\pm$ can be extended in a natural way to act on matrices of
Wiener algebra elements. Hence, according to \eqref{3.16} and \eqref{3.17}, for
each $x\in\R$ the Faddeev functions $m_{r,l}(x,\cdot)\in\CW_+^{n\times n}$.

 For $n=1$ the following result is most easily proved using the Gelfand theory
of commutative Banach algebras \cite{GRS}. The extension to arbitrary $n$ is
elementary, because the determinant and the cofactor matrix of an element of
$\CW^{n\times n}$ belong to $\CW$ and $\CW^{n\times n}$, respectively.

\begin{theorem}\label{th:3.2}
If for some complex $n\times n$ matrix $H_\infty$ and some
$H\in L^1(\R)^{n\times n}$ the Fourier transform
$H_\infty+\int_{-\infty}^\infty dz\,e^{ik z}H(z)$ is a nonsingular matrix for
every $k\in\R$ and if $\det H_\infty\neq0$, then there exists
$K\in L^1(\R)^{n\times n}$ such that
$$\left[H_\infty+\int_{-\infty}^\infty dz\,e^{ik z}H(z)\right]^{-1}
=(H_\infty)^{-1}+\int_{-\infty}^\infty dz\,e^{ik z}K(z)$$
for every $k\in\R$.
\end{theorem}

Using that $m_{r,l}(x,\cdot)\in\CW_+^{n\times n}$, we easily prove with the help
of \eqref{3.4} that $2ik[I_n-A_{r,l}(k)]$ belongs to $\CW_+^{n\times n}$ and
$2ikB_{r,l}(k)$ belongs to $\CW^{n\times n}$. Assuming the absence of spectral
singularities and to be in the generic case, we easily prove that the reflection
coefficients $R_{r,l}(k)$ belong to $\CW_0^{n\times n}$ and the transmission
coefficients $A_{r,l}(k)^{-1}$ to $\CW_+^{n\times n}$. Indeed,
$$\frac{k}{k+1}A_{r,l}(k)=\frac{k}{k+i}I_n-\frac{2ik[I_n-A_{r,l}(k)]}{2i(k+i)}$$
belongs to $\CW_+^{n\times n}$ and therefore, by Theorem \ref{th:3.2} and
$k/(k+i)\in\CW$,
$$A_{r,l}(k)^{-1}=\frac{k}{k+i}\left[\frac{k}{k+i}A_{r,l}(k)\right]^{-1}$$
belongs to $\CW^{n\times n}$. Analogously, since
$$\frac{k}{k+i}B_{r,l}(k)=\frac{2ikB_{r,l}(k)}{2i(k+i)}$$
belongs to $\CW^{n\times n}$, we easily see that
$$R_{r,l}(k)=\left[\frac{k}{k+i}B_{r,l}(k)\right]\left[\frac{k}{k+i}A_{r,l}(k)
\right]^{-1}$$
belongs to $\CW_0^{n\times n}$.

The exceptional case requires more attention. Assuming that the potential
$\bQ\in L^1(\R;(1+|x|)^2dx)$, we can differentiate \eqref{3.2} with respect to
$k$ and solve the resulting integral equations by iteration. Next, with the help
of \eqref{3.4} we obtain
$$A_{r,l}(k)+B_{r,l}(k)=I_n+\int_{-\infty}^\infty dy\,\frac{e^{\mp2iky}-1}{2ik}
\bQ(y)m_{r,l}(y,k).$$
This expression satisfies
$$\lim_{k\to0^\pm}[A_{r,l}(k)+B_{r,l}(k)]=I_n\mp\int_{-\infty}^\infty dy\,
y\bQ(y)m_{r,l}(y,0)=I_n\mp E_{r,l},$$
provided $\bQ\in L^1(\R;(1+|x|)^2dx)$. Taking the limit as $k\to0^\pm$ we get
\begin{subequations}\label{3.21}
\begin{align}
\lim_{\substack{k\to0\\k\in\C^+\cup\R}}\frac{2ikA_{r,l}(k)-\bzD_{r,l}}{2ik}
&=I_n-\frac{1}{2i}\int_{-\infty}^\infty dy\,\bQ(y)\dot{m}_{r,l}(y,0)\nonumber\\
&=I_n-\frac{1}{2i}G_{r,l},\label{3.21a}\\
\lim_{k\to0^\pm}\frac{2ikB_{r,l}(k)+\bzD_{r,l}}{2ik}&=\int_{-\infty}^\infty dy\,
\bQ(y)[\mp ym_{r,l}(y,0)+\tfrac{1}{2i}\dot{m}_{r,l}(y,0)]\nonumber\\
&=\mp E_{r,l}+\frac{1}{2i}G_{r,l},\label{3.21b}
\end{align}
\end{subequations}
where the overdot indicates differentiation with respect to $k$.

In the superexceptional case, where $\bzD_{r,l}=0_{n\times n}$, it is clear that
$A_{r,l}\in\CW_+^{n\times n}$, provided $\bQ\in L^1(\R;(1+|x|)^2dx)$. Assuming
the absence of spectral singularities, using the nonsingularity of $A_{r,l}(0)$,
and applying Theorem \ref{th:3.2}, we get the transmission coefficients
$A_{r,l}(k)^{-1}$ to belong to $\CW^{n\times n}$. Since $A_{r,l}(k)+B_{r,l}(k)$
belongs to $\CW^{n\times n}$, we conclude that
$$R_{r,l}(k)=[A_{r,l}(k)+B_{r,l}(k)]A_{r,l}(k)^{-1}-I_n\in\CW_0^{n\times n}.$$

At present it is not known if, under the absence of spectral singularities, the
reflection and transmission coefficients belong to $\CW^{n\times n}$ in any
other exceptional case and for general $\bQ\in L^1(\R;(1+|x|)dx)$. Under the
condition $\bQ\in L^1(\R;(1+|x|)dx)$, the continuity of the reflection and
transmission coefficients at $k=0$ is known for $n=1$ \cite{K88} and for
selfadjoint potentials \cite{AKV01}.

\section{Inverse scattering problem}\label{sec:4}

In this section we introduce the Marchenko integral equations for the matrix
Schr\"o\-din\-ger equation \eqref{1.5} with Faddeev class potential $\bQ$
satisfying \eqref{1.2}. We make use of the hypothesis that the reflection
coefficients $R_{r,l}\in\CW_0^{n\times n}$, something that we have not been able
to prove in the most general exceptional case. For the sake of simplicity we
assume that the poles of $A_{r,l}(k)^{-1}$ in $\C^+$ are simple. The extension
to multiple pole situations is rather technical but straightforward \cite{DM1}.
Inverse scattering theory is well documented in the scalar case
\cite{Fa64,DT79,CD82,CS89,AC}, as it is in the matrix half-line case
\cite{AM,AW18,AW20}. Without any details, the matrix full-line case is given in
\cite{Wd}.

Let us first write the {\it transmission coefficients} in the form
\begin{equation}\label{4.1}
A_r(k)^{-1}=A_{r0}(k)+\sum_{s=1}^N\,\frac{\tau_{r;s}}{k-k_s},\quad
A_l(k)^{-1}=A_{l0}(k)+\sum_{s=1}^N\,\frac{\tau_{l;s}}{k-k_s},
\end{equation}
where $k_1,\ldots,k_N$ are the distinct simple poles of $A_{r,l}(k)^{-1}$
in $\C^+$, $\tau_{r;s}$ and $\tau_{l;s}$ are the residues of $A_r(k)^{-1}$ and
$A_l(k)^{-1}$ at $k=k_s$ ($s=1,\ldots,N$), and $A_{r0}(k)$ and $A_{l0}(k)$ are
continuous in $k\in\C^+\cup\R$, are analytic in $k\in\C^+$, and tend to $I_n$ as
$k\to\infty$ from within $\C^+\cup\R$.

Letting $\tau_{r;s}$ and $\tau_{r;\os}$ (or: $\tau_{l;s}$ and $\tau_{l;\os}$)
stand for the residues of $A_r(k)^{-1}$ (or: $A_l(k)^{-1}$) at $k_s$ and
$-k_s^*$, we get with the help of \eqref{3.11}
\begin{align}
\tau_{r;\os}&=\lim_{k\to-k_s^*}\,(k+k_s^*)A_r(k)^{-1}
=\bzs_3\lim_{k\to-k_s^*}\,(k+k_s^*){A_l(-k^*)^\dagger}^{-1}\bzs_3\nonumber\\
&=-\bzs_3\left(\lim_{\zz\to k_s}\,(\zz-k_s)A_l(\zz)^{-1}\right)^\dagger\bzs_3
=-\bzs_3\tau_{l;s}^\dagger\bzs_3.\label{4.2}
\end{align}
Similarly, we get $\tau_{l;\os}=-\bzs_3\tau_{r;s}^\dagger\bzs_3$.

Next, let us consider the Riemann-Hilbert problem
\begin{equation}\label{4.3}
\begin{pmatrix}F_l(x,-k)&F_r(x,-k)\end{pmatrix}
=\begin{pmatrix}F_r(x,k)&F_l(x,k)\end{pmatrix}
\begin{pmatrix}A_r(k)^{-1}&-R_l(k)\\-R_r(k)&A_l(k)^{-1}\end{pmatrix},
\end{equation}
where $\det A_{r,l}(k)\neq0$ for $0\neq k\in\R$. In the generic and
superexceptional cases while assuming the absence of spectral singularities,
we write
\begin{equation}\label{4.4}
R_r(k)=\int_{-\infty}^\infty d\za\,e^{-ik\za}\hat{R}_r(\za),\qquad
R_l(k)=\int_{-\infty}^\infty d\za\,e^{ik\za}\hat{R}_l(\za),
\end{equation}
where the entries of the $n\times n$ matrix functions $\hat{R}_{r,l}(\za)$
belong to $L^1(\R)$. In fact, we only need to assume the absence of spectral
singularities as well as $R_{r,l}\in\CW_0^{n\times n}$ for the derivations below
to be valid. Rewriting the left half of \eqref{4.3} we obtain
\begin{align}
e^{ik x}&F_l(x,-k)=e^{ik x}F_r(x,k)A_r(k)^{-1}-e^{-ik x}F_l(x,k)
e^{2ik x}R_r(k)\nonumber\\&\!=\!e^{ik x}F_r(x,k)A_{r0}(k)+\sum_{s=1}^N\,
\frac{e^{ik x}F_r(x,k)-e^{ik_sx}F_r(x,k_s)}{k-k_s}\tau_{r;s}\nonumber\\
&+i\sum_{s=1}^N\,e^{2ik_sx}\frac{e^{-ik_sx}F_l(x,k_s)}{k-k_s}N_{r;s}
-e^{-ik x}F_l(x,k)e^{2ik x}R_r(k),\label{4.5}
\end{align}
where the {\it norming constants} $N_{r;s}$ ($s=1,\ldots,N$) are defined by
\begin{equation}\label{4.6}
 F_r(x,k_s)\tau_{r;s}=i\,F_l(x,k_s)N_{r;s},\qquad s=1,\ldots,N.
\end{equation}
By the same token, rewriting the right half of \eqref{4.3} we get
\begin{align}
e^{-ik x}&F_r(x,-k)=e^{-ik x}F_l(x,k)A_l(k)^{-1}-e^{ik x}F_r(x,k)
e^{-2ik x}R_l(k)\nonumber\\&\!=\!e^{-ik x}F_l(x,k)A_{l0}(k)+\sum_{s=1}^N\,
\frac{e^{-ik x}F_l(x,k)-e^{-ik_sx}F_l(x,k_s)}{k-k_s}\tau_{l;s}
\nonumber\\
&+i\sum_{s=1}^N\,e^{-2ik_s}\frac{e^{ik_sx}F_r(x,k_s)}{k-k_s}N_{l;s}
-e^{ik x}F_r(x,k)e^{-2ik x}R_l(k),\label{4.7}
\end{align}
where the norming constants $N_{l;s}$ ($s=1,\ldots,N$) are defined by
\begin{equation}\label{4.8}
 F_l(x,k_s)\tau_{l;s}=i\,F_r(x,k_s)N_{l;s},\qquad s=1,\ldots,N.
\end{equation}
 For every $x\in\R$ the identities \eqref{4.5} and \eqref{4.7} are equations in
$\CW^{n\times n}$. Using $\Pi_-$ to project these two equations onto
$\CW_-^{n\times n}$ along $\CW_{+,0}^{n\times n}$, we obtain
\begin{align*}
e^{ik x}F_l(x,-k)&=I_n+i\sum_{s=1}^N e^{2ik_sx}
\frac{e^{-ik_sx}F_l(x,k_s)}{k-k_s}N_{r;s}\nonumber\\
&-\Pi_-\left[e^{-ik x}F_l(x,k)e^{2ik x}R_r(k)\right],\\
e^{-ik x}F_r(x,-k)&=I_n+i\sum_{s=1}^N e^{-2ik_sx}
\frac{e^{ik_sx}F_r(x,k_s)}{k-k_s}N_{l;s}\nonumber\\
&-\Pi_-\left[e^{ik x}F_r(x,k)e^{-2ik x}R_l(k)\right]\!,
\end{align*}
respectively. Using \eqref{3.16}, \eqref{4.4}, and the identity
\begin{equation}\label{4.9}
\int_0^\infty dw\,e^{-ik w}e^{ik_sw}=\frac{-i}{k-k_s},\qquad
s=1,2,\ldots,N,
\end{equation}
we strip off the Fourier transforms and arrive at the Marchenko integral
equations
\begin{subequations}\label{4.10}
\begin{align}
K(x,y)+\zO_r(x+y)&+\int_x^\infty dz\,K(x,z)\zO_r(z+y)=0_{n\times n},
\label{4.10a}\\
J(x,y)+\zO_l(x+y)&+\int_{-\infty}^x dz\,J(x,z)\zO_l(z+y)=0_{n\times n},
\label{4.10b}
\end{align}
\end{subequations}
where the {\it Marchenko integral kernels} are given by
\begin{subequations}\label{4.11}
\begin{align}
\zO_r(w)&=\hat{R}_r(w)+\sum_{s=1}^N\,e^{ik_sw}N_{r;s},\label{4.11a}\\
\zO_l(w)&=\hat{R}_l(w)+\sum_{s=1}^N\,e^{-ik_sw}N_{l;s}.\label{4.11b}
\end{align}
\end{subequations}

Indeed, to prove \eqref{4.10a} and \eqref{4.11a} we compute
\begin{align*}
&\int_0^\infty d\za\,e^{-ik\za}K(x,x+\za)\\&=
-\int_0^\infty d\za\,e^{-ik\za}\sum_{s=1}^N\left[e^{ik_s(\za+2x)}N_{r;s}
+\int_0^\infty d\zb\,e^{ik_s(\za+\zb+2x)}K(x,x+\zb)N_{r;s}\right]\\
&-\int_0^\infty d\za\,e^{-ik\za}\left[\hat{R}_r(\za+2x)+\int_0^\infty d\zb\,
K(x,x+\zb)\hat{R}_r(\za+\zb+2x)\right].
\end{align*}
Stripping off the Fourier transform and putting $y=x+\za$ and $z=x+\zb$, we get
\eqref{4.10a} and \eqref{4.11a}. Analogously, to prove \eqref{4.10b} and
\eqref{4.11b} we compute
\begin{align*}
&\int_0^\infty d\za\,e^{-ik\za}J(x,x-\za)\\&=-\int_0^\infty d\za\,e^{-ik\za}
\sum_{s=1}^N\left[e^{-ik_s(2x-\za)}N_{l;s}\!+\!\int_0^\infty d\zb\,
e^{-ik_s(2x-\za-\zb)}J(x,x-\zb)N_{l;s}\right]\\
&-\int_0^\infty d\za\,e^{-ik\za}\left[\hat{R}_l(2x-\za)+\int_0^\infty d\zb\,
J(x,x-\zb)\hat{R}_l(2x-\za-\zb)\right].
\end{align*}
Stripping off the Fourier transform and putting $y=x-\za$ and $z=x-\zb$, we get
\eqref{4.10b}.

As in \cite{DM1}, we can prove the adjoint symmetry relations
\begin{equation}\label{4.12}
\zO_{r,l}(w)=\bzs_3\zO_{r,l}(w)^\dagger\bzs_3,
\end{equation}
thus implying the following symmetry relations for the norming constants:
\begin{equation}\label{4.13}
N_{r;s}=\bzs_3N_{r;s}^\dagger\bzs_3,\qquad
N_{l;s}=\bzs_3N_{l;s}^\dagger\bzs_3.
\end{equation}
 For the rather tedious details we refer to Appendix~\ref{sec:B}.

\section{AKNS vs. matrix Schr\"o\-din\-ger data}\label{sec:5}

Let us define the Jost solutions of the AKNS system \eqref{1.3} as those
solutions of \eqref{1.3} that for $k\in\R$ satisfy the asymptotic
relations
\begin{subequations}\label{5.1}
\begin{alignat}{3}
\F_l(x,k)&=e^{-ikx\bzs_3}[I_n+o(1)],&\qquad&x\to+\infty,\label{5.1a}\\
\F_r(x,k)&=e^{-ikx\bzs_3}[I_n+o(1)],&\qquad&x\to-\infty.\label{5.1b}
\end{alignat}
\end{subequations}
Since we are dealing with a first order system, there exist nonsingular matrices
$\A_l(k)$ and $\A_r(k)$ such that
\begin{equation}\label{5.2}
\F_r(x,k)=\F_l(x,k)\A_r(k),\qquad\F_l(x,k)=\F_r(x,k)\A_l(k),
\end{equation}
and hence $\A_l(k)$ and $\A_r(k)$ are each other's inverses. Equations
\eqref{5.1} and \eqref{5.2} imply the asymptotic relations
\begin{subequations}\label{5.3}
\begin{alignat}{3}
\F_l(x,k)&=e^{-ikx\bzs_3}\A_l(k)[I_n+o(1)],&\qquad&x\to-\infty,\label{5.3a}\\
\F_r(x,k)&=e^{-ikx\bzs_3}\A_r(k)[I_n+o(1)],&\qquad&x\to+\infty.\label{5.3b}
\end{alignat}
\end{subequations}
Since the matrix $-ik\bzs_3+\CQ$ in \eqref{1.3} has trace $-ik(m_1-m_2)$, the
determinants of $\F_l(x,k)$ and $\F_r(x,k)$ equal $e^{-ik(m_1-m_2)x}$. As
a result, $\det\A_{r,l}(k)=1$.

Let us relate the AKNS and matrix Schr\"o\-din\-ger Jost solutions. Putting
$$\be_1=I_{m_1}\oplus 0_{m_2\times m_2},\qquad
\be_2=0_{m_1\times m_1}\oplus I_{m_2},$$
we obtain from the principal asymptotic relations
\begin{subequations}\label{5.4}
\begin{align}
\F_l(x,k)&=F_l(x,-k)\be_1+F_l(x,k)\be_2,\label{5.4a}\\
\F_r(x,k)&=F_r(x,k)\be_1+F_r(x,-k)\be_2,\label{5.4b}
\end{align}
\end{subequations}
as well as
\begin{subequations}\label{5.5}
\begin{align}
 F_l(x,k)&=\F_l(x,-k)\be_1+\F_l(x,k)\be_2,\label{5.5a}\\
 F_r(x,k)&=\F_r(x,k)\be_1+\F_r(x,-k)\be_2.\label{5.5b}
\end{align}
\end{subequations}
Using $\be_1\be_2=\be_2\be_1=0_{n\times n}$ and $\be_1^2+\be_2^2=I_n$ it is easy
to derive \eqref{5.5} from \eqref{5.4} and vice versa. Using
\eqref{5.4} we obtain the triangular representations for the AKNS system
\begin{subequations}\label{5.6}
\begin{align}
\F_l(x,k)&=e^{-ikx\bzs_3}+\int_x^\infty dy\,K(x,y)e^{-iky\bzs_3},\label{5.6a}\\
\F_r(x,k)&=e^{-ikx\bzs_3}+\int_{-\infty}^x dy\,J(x,y)e^{-iky\bzs_3},
\label{5.6b}
\end{align}
\end{subequations}
where the kernel functions $K$ and $J$ coincide with those for the matrix
Schr\"o\-din\-ger system.

Let us now relate the AKNS and matrix Schr\"o\-din\-ger scattering coefficients.
Let us introduce the standard partitioning
\begin{equation}\label{5.7}
H=\begin{pmatrix}H_1&H_2\\H_3&H_4\end{pmatrix}
\end{equation}
of an $n\times n$ matrix $H$ into the $m_1\times m_1$ matrix $H_1$, the
$m_1\times m_2$ matrix $H_2$, the $m_2\times m_1$ matrix $H_3$, and the
$m_2\times m_2$ matrix $H_4$. Using the asymptotic relations \eqref{3.3} and
\eqref{5.3} we obtain for $k\in\R$
\begin{align*}
e^{ikx}A_l(k)+e^{-ikx}B_l(k)&=\begin{pmatrix}e^{ikx}\A_{l1}(-k)
&e^{-ikx}\A_{l2}(k)\\e^{-ikx}\A_{l3}(-k)&e^{ikx}\A_{l4}(k)\end{pmatrix},\\
e^{-ikx}A_r(k)+e^{ikx}B_r(k)
&=\begin{pmatrix}e^{-ikx}\A_{r1}(k)&e^{ikx}\A_{r2}(-k)\\
e^{ikx}\A_{r3}(k)&e^{-ikx}\A_{r4}(-k)\end{pmatrix}.
\end{align*}
Consequently, for $0\neq k\in\R$ we have
\begin{subequations}\label{5.8}
\begin{alignat}{3}
A_l(k)&=\begin{pmatrix}\A_{l1}(-k)&0_{m_1\times m_2}\\
0_{m_2\times m_1}&\A_{l4}(k)\end{pmatrix},&\quad
A_r(k)&=\begin{pmatrix}\A_{r1}(k)&0_{m_1\times m_2}\\
0_{m_2\times m_1}&\A_{r4}(-k)\end{pmatrix},\label{5.8a}\\
B_l(k)&=\begin{pmatrix}0_{m_1\times m_1}&\A_{l2}(k)\\
\A_{l3}(-k)&0_{m_2\times m_2}\end{pmatrix},&\quad
B_r(k)&=\begin{pmatrix}0_{m_1\times m_1}&\A_{r2}(-k)\\
\A_{r3}(k)&0_{m_2\times m_2}\end{pmatrix}.\label{5.8b}
\end{alignat}
\end{subequations}
Conversely, for $0\neq k\in\R$ we get in terms of blocks of $A_{r,l}(k)$ and
$B_{r,l}(k)$
\begin{subequations}\label{5.9}
\begin{align}
\A_l(k)&=\begin{pmatrix}A_{l1}(k)&B_{l2}(k)\\
B_{l3}(-k)&A_{l4}(k)\end{pmatrix},\label{5.9a}\\
\A_r(k)&=\begin{pmatrix}A_{r1}(k)&B_{r2}(-k)\\
B_{r3}(k)&A_{r4}(-k)\end{pmatrix}.\label{5.9b}
\end{align}
\end{subequations}
Since $\A_l(k)$ and $\A_r(k)$ are continuous in $k\in\R$ and tend to $I_n$ as
$k\to\pm\infty$, we see that $A_l(k)$, $B_l(k)$, $A_r(k)$, and $B_r(k)$ are
continuous in $k\in\R$ and have finite limits as $k\to\pm\infty$. Of course,
the block entries of $\A_{l,r}(k)$ have all the usual continuity and analyticity
properties, also at $k=0$. Thus, we are always in the exceptional case. We are
in the superexceptional case iff $k=0$ is not a spectral singularity of the AKNS
system \eqref{1.3}. Further, the matrix Schr\"o\-din\-ger reflection
coefficients are given by
\begin{subequations}\label{5.10}
\begin{align}
R_l(k)&=\begin{pmatrix}0_{m_1\times m_1}&\R_{l1}(k)\\
\R_{l4}(k)&0_{m_2\times m_2}\end{pmatrix}
=\begin{pmatrix}0_{m_1\times m_1}&\A_{l2}(k)\A_{l4}(k)^{-1}\\
\A_{l3}(-k)\A_{l1}(-k)^{-1}&0_{m_2\times m_2}\end{pmatrix},\label{5.10a}\\
R_r(k)&=\begin{pmatrix}0_{m_1\times m_1}&\R_{r4}(k)\\
\R_{r1}(k)&0_{m_2\times m_2}\end{pmatrix}
=\begin{pmatrix}0_{m_1\times m_1}&\A_{r2}(-k)\A_{r4}(-k)^{-1}\\
\A_{r3}(k)\A_{r1}(k)^{-1}&0_{m_2\times m_2}\end{pmatrix}.\label{5.10b}
\end{align}
\end{subequations}
 Finally, the AKNS system is reflectionless iff the matrix Schr\"o\-din\-ger
system is reflectionless.

Let $\tau_{r;s}$ and $\tau_{l;s}$ be the residues of $A_r(k)^{-1}$ and
$A_l(k)^{-1}$ at the simple pole $k_s\in\C^+$. Now let $\tau_{l4;s}$ and
$\tau_{r1;s}$ be the residues of $\A_{l4}(k)^{-1}$ and $\A_{r1}(k)^{-1}$ at
$k=k_s$ and $\tau_{l1;s}$ and $\tau_{r4;s}$ the residues of $\A_{l1}(k)^{-1}$
and $\A_{r4}(k)^{-1}$ at the simple pole $k=-k_s$. Then \eqref{5.8a} implies
that
$$\tau_{r;s}=\tau_{r1;s}\oplus(-\tau_{r4;s}),\qquad
\tau_{l;s}=(-\tau_{l1;s})\oplus\tau_{l4;s}.$$
Let us define the AKNS norming constants as follows:
\begin{subequations}\label{5.11}
\begin{align}
\F_r(x,k_s)\be_1\tau_{r1;s}&=i\,\F_l(x,k_s)\be_2\N_{r1;s},\label{5.11a}\\
\F_r(x,-k_s)\be_2\tau_{r4;s}&=-i\,\F_l(x,-k_s)\be_1\N_{r4;s},\label{5.11b}\\
\F_l(x,k_s)\be_2\tau_{l4;s}&=i\,\F_r(x,k_s)\be_1\N_{l4;s},\label{5.11c}\\
\F_l(x,-k_s)\be_1\tau_{l1;s}&=-i\,\F_r(x,-k_s)\be_2\N_{l1;s}.\label{5.11d}
\end{align}
\end{subequations}
We then obtain
\begin{alignat}{3}
N_{r;s}&=\begin{pmatrix}0_{m_1\times m_1}&\N_{r4;s}\\
\N_{r1;s}&0_{m_2\times m_2}\end{pmatrix},\qquad
N_{l;s}&=\begin{pmatrix}0_{m_1\times m_1}&\N_{l1;s}\\
\N_{l4;s}&0_{m_2\times m_2}\end{pmatrix}.\label{5.12}
\end{alignat}

\section{Time evolution}\label{sec:6}

In this section we establish the time evolution of the scattering data of the
matrix Schr\"o\-din\-ger equation and the accompanying AKNS system. We then go
on to derive, in either case, the Marchenko integral kernels as a function of
time. These results allow us, in the next section, to derive the reflectionless
solutions of the nonlocal integrable equation \eqref{2.3} and hence of the
focusing matrix NLS equation.

\subsection{Matrix Schr\"odinger time evolution}\label{sec:6.1}

The Lax pair $\{\CL,A\}$ for the matrix Schr\"o\-din\-ger equation \eqref{2.3}
is given by \eqref{2.1}, where $\CB_0=\bQ$. Equation \eqref{2.3} is compatible
with the linear system
$$\CL v=k^2v,\qquad v_t=Av,$$
where $\CL=-\pt_x^2+\bQ$. We may therefore write
\begin{align*}
v_t&=Av=2i\bzs_3v_{xx}-2i\bzs_3\CQ v_x-i\bzs_3\bQ v\nonumber\\
&=2i\bzs_3(\bQ-k^2\bone)v-2i\bzs_3\CQ v_x-i\bzs_3\bQ v\nonumber\\
&=i\bzs_3\left\{(\bQ-2k^2\bone)v-2\CQ v_x\right\},
\end{align*}
where $\bone$ stands for the identity operator on a suitable function space.
Next, we compute
\begin{align*}
(v_x)_t&=(Av)_x=i\bzs_3\left((\bQ-2k^2\bone)v_x+\bQ_xv-2\CQ_xv_x
-2\CQ(\bQ-k^2\bone)v\right)\nonumber\\
&=i\bzs_3\left(\bQ_x-2\CQ\bQ+2k^2\CQ\right)v
+i\bzs_3(\bQ-2k^2\bone-2\CQ_x)v_x.
\end{align*}
Hence, putting $\bV=\left(\begin{smallmatrix}v\\v_x\end{smallmatrix}\right)$
we get the linear system
\begin{equation}\label{6.1}
\bV_x=\bX(x,k;t)\bV,\qquad\bV_t=\bT(x,k;t)\bV,
\end{equation}
where
\begin{subequations}\label{6.2}
\begin{align}
\bX(x,k;t)&=\begin{pmatrix}0_{n\times n}&I_n\\ \bQ(x;t)-k^2I_n&0_{n\times n}
\end{pmatrix},\label{6.2a}\\
\bT(x,k;t)&=\begin{pmatrix}i\bzs_3(\bQ-2k^2I_n)&-2i\bzs_3\CQ\\
i\bzs_3(\bQ_x-2\CQ\bQ+2k^2\CQ)&i\bzs_3(\bQ-2k^2I_n-2\CQ_x)
\end{pmatrix}.\label{6.2b}
\end{align}
\end{subequations}
Then we easily compute
$$i(\bzs_3\oplus\bzs_3)\left(\bX_t-\bT_x+\bX\bT-\bT\bX\right)
=\begin{pmatrix}0_{n\times n}&0_{n\times n}\\
\bE_{21}&0_{n\times n}\end{pmatrix},$$
where
\begin{align*}
\bE_{21}&=i\bzs_3\bQ_t+\bQ_{xx}-2(\CQ\bQ)_x+2k^2\CQ_x\\
&-\bzs_3(\bQ-k^2I_n)\bzs_3(\bQ-2k^2I_n)
+(\bQ-2k^2I_n-2\CQ_x)(\bQ-k^2I_n)\\
&=i\bzs_3\bQ_t+\bQ_{xx}+\bQ^2-\bzs_3\bQ\bzs_3\bQ-2\CQ\bQ_x-4\CQ_x\bQ\\
&+k^2\left(2\CQ_x+2\bzs_3\bQ\bzs_3+\bQ-\bQ-2\bQ+2\CQ_x\right)\\
&=i\bzs_3\bQ_t+\bQ_{xx}+\bQ^2-\bzs_3\bQ\bzs_3\bQ-2\CQ\bQ_x
-2\left(\bQ^2-\bzs_3\bQ\bzs_3\bQ\right)\\
&=i\bzs_3\bQ_t+\bQ_{xx}-\bQ^2+\bzs_3\bQ\bzs_3\bQ-2\CQ\bQ_x,
\end{align*}
as claimed. Thus the zero curvature condition for the AKNS pair $\{\bX,\bT\}$
is equivalent to the nonlinear evolution equation \eqref{2.3}.

There exist nonsingular matrices $C_{\bF_l}(k;t)$ and $C_{\bF_r}(k;t)$ not
depending on $k\in\R$ such that
$$\bF_l(x,k;t)=\bV(x,k;t)C_{\bF_l}(k;t)^{-1},\qquad
\bF_r(x,k;t)=\bV(x,k;t)C_{\bF_r}(k;t)^{-1}.$$
Then a simple differentiation yields
\begin{align*}
\left[C_{\bF_l}(k;t)\right]_tC_{\bF_l}(k;t)^{-1}&=\bF_l^{-1}\bT\bF_l
-\bF_l^{-1}[\bF_l]_t,\\
\left[C_{\bF_r}(k;t)\right]_tC_{\bF_r}(k;t)^{-1}&=\bF_r^{-1}\bT\bF_r
-\bF_r^{-1}[\bF_r]_t,
\end{align*}
where the two left-hand sides do not depend on $x\in\R$. We may therefore
compute the $x\to\pm\infty$ limits of the two right-hand sides and obtain
$$\left[C_{\bF_l}(k;t)\right]_tC_{\bF_l}(k;t)^{-1}
=\left[C_{\bF_r}(k;t)\right]_tC_{\bF_r}(k;t)^{-1}=-2ik^2(\bzs_3\oplus\bzs_3).$$
Consequently,
\begin{subequations}\label{6.3}
\begin{align}
C_{\bF_l}(k;t)&=\left[e^{-2ik^2t\bzs_3}\oplus e^{-2ik^2t\bzs_3}\right]
C_{\bF_l}(k;0),\label{6.3a}\\
C_{\bF_r}(k;t)&=\left[e^{-2ik^2t\bzs_3}\oplus e^{-2ik^2t\bzs_3}\right]
C_{\bF_r}(k;0).\label{6.3b}
\end{align}
\end{subequations}
Relating $\bF_{r,l}(x,k;t)$ by means of the equalities
$$\bF_r(x,k;t)=\bF_l(x,k;t)\bA_r(k;t),\qquad
\bF_l(x,k;t)=\bF_r(x,k;t)\bA_l(k;t),$$
where the factors $\bA_{r,l}(k;t)$ are given by the matrices [cf. \eqref{3.5}]
$$\bA_r(k;t)=\begin{pmatrix}A_r(k;t)&B_r(-k;t)\\B_r(k;t)&A_r(-k;t)\end{pmatrix},
\quad
\bA_l(k;t)=\begin{pmatrix}A_l(-k;t)&B_l(k;t)\\B_l(-k;t)&A_l(k;t)\end{pmatrix},$$
we compute
\begin{align*}
[\bA_r]_t&=-\bF_l^{-1}[\bF_l]_t\bF_l^{-1}\bF_r+\bF_l^{-1}[\bF_r]_t\\
&=-\bF_l^{-1}\left(\bT\bF_l
-\bF_l\left[C_{\bF_l}(k;t)\right]_tC_{\bF_l}(k;t)^{-1}\right)\bA_r\\
&+\bF_l^{-1}\left(\bT\bF_r
-\bF_r\left[C_{\bF_r}(k;t)\right]_tC_{\bF_r}(k;t)^{-1}\right)\\
&=\left[C_{\bF_l}(k;t)\right]_tC_{\bF_l}(k;t)^{-1}\bA_r
-\bA_r\left[C_{\bF_l}(k;t)\right]_tC_{\bF_l}(k;t)^{-1}\\
&=2ik^2\left(\bA_r(\bzs_3\oplus\bzs_3)-(\bzs_3\oplus\bzs_3)\bA_r\right).
\end{align*}
Using that $\bA_{r,l}(k;t)$ are each other's inverses, we get
$$\bA_{r,l}(k;t)=\left[e^{-2ik^2t\bzs_3}\oplus e^{-2ik^2t\bzs_3}\right]
\bA_{r,l}(k;0)\left[e^{2ik^2t\bzs_3}\oplus e^{2ik^2t\bzs_3}\right].$$
Therefore,
\begin{subequations}\label{6.4}
\begin{align}
A_{r,l}(k;t)&=e^{-2ik^2t\bzs_3}A_{r,l}(k;0)e^{2ik^2t\bzs_3},\label{6.4a}\\
B_{r,l}(k;t)&=e^{-2ik^2t\bzs_3}B_{r,l}(k;0)e^{2ik^2t\bzs_3}.\label{6.4b}
\end{align}
\end{subequations}
 For the reflection coefficients we get
\begin{equation}\label{6.5}
R_{r,l}(k;t)=e^{-2ik^2t\bzs_3}R_{r,l}(k;0)e^{2ik^2t\bzs_3}.
\end{equation}
Moreover, if $k_s$ is a simple pole of $A_{r,l}(k;t)$ in $\C^+$, \eqref{6.4a}
implies that
$$\tau_{r;s}(t)=e^{-2ik_s^2t\bzs_3}\tau_{r;s}(0)e^{2ik_s^2t\bzs_3},\qquad
\tau_{l;s}(t)=e^{-2ik_s^2t\bzs_3}\tau_{l;s}(0)e^{2ik_s^2t\bzs_3}.$$
As a result,
\begin{equation}\label{6.6}
[\tau_{r;s}]_t=2ik_s^2\left(\tau_{r;s}\bzs_3-\bzs_3\tau_{r;s}\right),\qquad
[\tau_{l;s}]_t=2ik_s^2\left(\tau_{l;s}\bzs_3-\bzs_3\tau_{l;s}\right).
\end{equation}

Let us now study the time evolution of the norming constants corresponding to
a simple pole $k_s$ of $A_r(k;t)^{-1}$ in $\C^+$. Equations \eqref{6.3} imply
\begin{equation}\label{6.7}
[\bF_{r,l}]_t=\bT\bF_{r,l}+2ik^2\bF_{r,l}(\bzs_3\oplus\bzs_3).
\end{equation}
Now recall \eqref{4.6} and \eqref{4.8} in the form
\begin{subequations}\label{6.8}
\begin{align}
 F_r(x,k_s;t)\tau_{r;s}&=iF_l(x,k_s;t)N_{r;s}(t),\label{6.8a}\\
 F_l(x,k_s;t)\tau_{l;s}&=iF_r(x,k_s;t)N_{l;s}(t).\label{6.8b}
\end{align}
\end{subequations}
Using the right/left upper block of \eqref{6.7} for $k=k_s$ to differentiate
\eqref{6.8} with respect to $t$ for $k=k_s$ and then substituting \eqref{6.6}
and \eqref{6.8}, we get
\begin{align*}
 F_l(x,k_s;t)\left\{[N_{r;s}]_t-2ik_s^2\left(N_{r;s}\bzs_3-\bzs_3N_{r;s}
\right)\right\}&=0_{n\times n},\\
 F_r(x,k_s;t)\left\{[N_{l;s}]_t-2ik_s^2\left(N_{l;s}\bzs_3-\bzs_3N_{l;s}
\right)\right\}&=0_{n\times n}.
\end{align*}
Using \eqref{3.1} we obtain
\begin{align*}
[N_{r;s}]_t&=2ik_s^2\left(N_{r;s}(t)\bzs_3-\bzs_3N_{r;s}(t)\right),\\
[N_{l;s}]_t&=2ik_s^2\left(N_{l;s}(t)\bzs_3-\bzs_3N_{l;s}(t)\right).
\end{align*}
Consequently,
\begin{subequations}\label{6.9}
\begin{align}
N_{r;s}(t)&=e^{-2ik_s^2t\bzs_3}N_{r;s}(0)e^{2ik_s^2t\bzs_3},\label{6.9a}\\
N_{l;s}(t)&=e^{-2ik_s^2t\bzs_3}N_{l;s}(0)e^{2ik_s^2t\bzs_3}.\label{6.9b}
\end{align}
\end{subequations}

\subsection{AKNS time evolution}\label{sec:6.2}

Equations \eqref{6.4a} and \eqref{5.8a}, applied for $k\in\R$ and then extended
by analytic continuation, imply that $\A_{l4}(k;t)$ and $\A_{r1}(k;t)$ for
$k\in\C^+\cup\R$ and $\A_{l1}(k;t)$ and $\A_{r4}(k;t)$ for $k\in\C^-\cup\R$ are
time independent. Equations \eqref{6.4b} and \eqref{5.8b} imply
\begin{subequations}\label{6.10}
\begin{alignat}{3}
\A_{l2}(k;t)&=e^{-4ik^2t}\A_{l2}(k;0),&\qquad
\A_{l3}(k;t)&=e^{4ik^2t}\A_{l3}(k;0),\label{6.10a}\\
\A_{r2}(k;t)&=e^{-4ik^2t}\A_{r2}(k;0),&\qquad
\A_{r3}(k;t)&=e^{4ik^2t}\A_{r3}(k;0),\label{6.10b}
\end{alignat}
\end{subequations}
where $k\in\R$. Equations \eqref{6.5} and \eqref{5.10} imply
\begin{subequations}\label{6.11}
\begin{alignat}{3}
\R_{l1}(k;t)&=e^{-4ik^2t}\R_{l1}(k;0),&\qquad
\R_{l4}(k;t)&=e^{4ik^2t}\R_{l4}(k;0),\label{6.11a}\\
\R_{r1}(k;t)&=e^{4ik^2t}\R_{r1}(k;0),&\qquad
\R_{r4}(k;t)&=e^{-4ik^2t}\R_{r4}(k;0),\label{6.11b}
\end{alignat}
\end{subequations}
where $k\in\R$. Finally, \eqref{6.9} and \eqref{5.12} imply
\begin{subequations}\label{6.12}
\begin{alignat}{3}
\N_{l1;s}(t)&=e^{-4ik_s^2t}\N_{l1;s}(0),&\qquad
\N_{l4;s}(t)&=e^{4ik_s^2t}\N_{l4;s}(0),\label{6.12a}\\
\N_{r1;s}(t)&=e^{4ik_s^2t}\N_{r1;s}(0),&\qquad
\N_{r4;s}(t)&=e^{-4ik_s^2t}\N_{r4;s}(0),\label{6.12b}
\end{alignat}
\end{subequations}
where $k_s$ is a simple pole of $A_{r,l}(k;t)^{-1}$ in $\C^+$.

\subsection{Marchenko kernels and time evolution}\label{sec:6.3}

The Marchenko integral kernels of the matrix Schr\"o\-din\-ger system can be
expressed in terms of those of the AKNS system as follows:
\begin{equation}\label{6.13}
\zO_r(w;t)=\begin{pmatrix}0_{m_1\times m_1}&\zO_{r4}(w;t)\\ \zO_{r1}(w;t)
&0_{m_2\times m_2}\end{pmatrix},\quad
\zO_l(w;t)=\begin{pmatrix}0_{m_1\times m_1}&\zO_{l1}(w;t)\\ \zO_{l4}(w;t)
&0_{m_2\times m_2}\end{pmatrix},
\end{equation}
where in the case of simple poles
\begin{subequations}\label{6.14}
\begin{align}
\zO_{r1}(w;t)&=\hat{R}_{r1}(w;t)
+\sum_{s=1}^N\,e^{ik_sw}e^{4ik_s^2t}N_{r1;s}(0),\label{6.14a}\\
\zO_{r4}(w;t)&=\hat{R}_{r4}(w;t)
+\sum_{s=1}^N\,e^{ik_sw}e^{-4ik_s^2t}N_{r4;s}(0),\label{6.14b}\\
\zO_{l1}(w;t)&=\hat{R}_{l1}(w;t)
+\sum_{s=1}^N\,e^{-ik_sw}e^{-4ik_s^2t}N_{l1;s}(0),\label{6.14c}
\end{align}
\begin{align}
\zO_{l4}(w;t)&=\hat{R}_{l4}(w;t)
+\sum_{s=1}^N\,e^{-ik_sw}e^{4ik_s^2t}N_{l4;s}(0).\label{6.14d}
\end{align}
\end{subequations}
Consequently, for simple poles we obtain
\begin{subequations}\label{6.15}
\begin{align}
\zO_r(w;t)&=\hat{R}_r(w;t)+\sum_{s=1}^N\,e^{ik_sw}e^{-4ik_s^2t\bzs_3}N_{r;s}(0),
\label{6.15a}\\
\zO_l(w;t)&=\hat{R}_l(w;t)+\sum_{s=1}^N\,e^{-ik_sw}e^{-4ik_s^2t\bzs_3}
N_{l;s}(0),\label{6.15b}
\end{align}
\end{subequations}
which we could also have obtained directly from \eqref{4.11} and \eqref{6.13}.
As a result, we have arrived at the linear PDEs
\begin{equation}\label{6.16}
[\zO_r]_t=4i\bzs_3[\zO_r]_{\za\za},\qquad
[\zO_l]_t=4i\bzs_3[\zO_l]_{\za\za},
\end{equation}
in accordance with \eqref{6.6}.

Let us discuss the adjoint symmetry relations of the Marchenko integral kernels.
According to Corollary \ref{th:B.4}, those of the matrix Schr\"o\-din\-ger
equation read
$$\zO_r(w;t)^\dagger=\bzs_3\zO_r(w;t)\bzs_3,\qquad
\zO_l(w;t)^\dagger=\bzs_3\zO_l(w;t)\bzs_3.$$
Using \eqref{6.13} this yields the adjoint symmetry relations for the AKNS
system
\begin{subequations}\label{6.17}
\begin{alignat}{3}
\zO_{r1}(w;t)^\dagger&=-\zO_{r4}(w;t),&\qquad
\zO_{r4}(w;t)^\dagger&=-\zO_{r1}(w;t),\label{6.17a}\\
\zO_{l1}(w;t)^\dagger&=-\zO_{l4}(w;t),&\qquad
\zO_{l4}(w;t)^\dagger&=-\zO_{l1}(w;t).\label{6.17b}
\end{alignat}
\end{subequations}

\section{Multisoliton solutions}\label{sec:7}

In this section we apply the matrix triplet method to write the Marchenko
integral kernels in separated form and solve them by separation of variables.
This method has been successfully applied to the Korteweg-de Vries (KdV)
equation \cite{AdC,AV1}, the NLS equation \cite{ADM,DM2}, the sine-Gordon
equation \cite{Sb2,ADM2}, the modified Korteweg-de Vries (mKdV) equation
\cite{D2}, and the Toda lattice equation \cite{Sb1}. An introduction to this
method can be found in \cite{SIMAI}.

In the reflectionless case we can write the Marchenko integral kernels of the
AKNS system as follows:
\begin{subequations}\label{7.1}
\begin{alignat}{3}
\zO_{r1}(w;t)&=\CC_re^{-w\CA_r}e^{-4it\CA_r^2}\CB_r,&\qquad\zO_{r4}(w;t)
&=-\CB_r^\dagger e^{-w\CA_r^\dagger}e^{4it{\CA_r^\dagger}^2}\CC_r^\dagger,
\label{7.1a}\\
\zO_{l1}(w;t)&=\CC_le^{w\CA_l}e^{4it\CA_l^2}\CB_l,&\qquad\zO_{l4}(w;t)
&=-\CB_l^\dagger e^{w\CA_l^\dagger}e^{-4it{\CA_l^\dagger}^2}\CC_l^\dagger,
\label{7.1b}
\end{alignat}
\end{subequations}
where we have taken into account the adjoint symmetry relations \eqref{6.17}.
Here $(\CA_r,\CB_r,\CC_r)$ and $(\CA_l,\CB_l,\CC_l)$ are matrix triplets, where all of the
eigenvalues of $\CA_r$ and $\CA_l$ have a positive real part. Let us now define
\begin{alignat*}{4}
\bA_r&=\begin{pmatrix}\CA_r^\dagger&0_{p\times p}\\0_{p\times p}&\CA_r
\end{pmatrix},&\quad
\bB_r&=\begin{pmatrix}0_{p\times m_1}&-\CC_r^\dagger\\ \CB_r&0_{p\times m_2}
\end{pmatrix},&\quad
\bC_r&=\begin{pmatrix}\CB_r^\dagger&0_{m_1\times p}\\0_{m_2\times p}&\CC_r
\end{pmatrix},\\
\bA_l&=\begin{pmatrix}\CA_l&0_{p\times p}\\0_{p\times p}&\CA_l^\dagger
\end{pmatrix},&\quad
\bB_l&=\begin{pmatrix}0_{p\times m_1}&\CB_l\\-\CC_l^\dagger&0_{p\times m_2}
\end{pmatrix},&\quad
\bC_l&=\begin{pmatrix}\CC_l&0_{m_1\times p}\\0_{m_2\times p}&\CB_l^\dagger
\end{pmatrix},
\end{alignat*}
where $\CA_r$ and $\CA_l$ are both matrices of order $p$ whose eigenvalues have
a positive real part. Putting $\tzs_3=\zs_3\otimes I_p=I_p\oplus(-I_p)$,
we obtain
\begin{subequations}\label{7.2}
\begin{align}
\zO_r(w;t)&=\bC_re^{-w\bA_r}e^{4it\tzs_3\bA_r^2}\bB_r,\label{7.2a}\\
\zO_l(w;t)&=\bC_le^{w\bA_l}e^{4it\tzs_3\bA_l^2}\bB_l,\label{7.2b}
\end{align}
\end{subequations}
where
\begin{equation}\label{7.3}
\tzs_3\bA_{r,l}=\bA_{r,l}\tzs_3,\quad
\bB_{r,l}\bzs_3=-\tzs_3\bB_{r,l},\quad\bzs_3\bC_{r,l}=\bC_{r,l}\tzs_3.
\end{equation}

Solving the Marchenko integral equations we get
\begin{subequations}\label{7.4}
\begin{align}
K(x,y,t)&=-\bW_r(x;t)e^{-y\bA_r}e^{4it\tzs_3\bA_r^2}\bB_r,\label{7.4a}\\
J(x,y,t)&=-\bW_l(x;t)e^{y\bA_l}e^{4it\tzs_3\bA_l^2}\bB_l,\label{7.4b}
\end{align}
\end{subequations}
where
\begin{align*}
\bW_r(x;t)&=\bC_re^{-x\bA_r}+\int_x^\infty dz\,K(x,z,t)\bC_re^{-z\bA_r},\\
\bW_l(x;t)&=\bC_le^{x\bA_l}+\int_{-\infty}^x dz\,J(x,z,t)\bC_le^{z\bA_l}.
\end{align*}
Substituting \eqref{7.4} into \eqref{4.10} and solving \eqref{4.10} for
$\bW_{r,l}(x;t)$ we get
\begin{align*}
\bW_r(x;t)&=\bC_re^{-x\bA_r}\left[I_{2p}+e^{-x\bA_r}e^{4it\tzs_3\bA_r^2}\bP_r
e^{-x\bA_r}\right]^{-1},\\
\bW_l(x;t)&=\bC_le^{x\bA_l}\left[I_{2p}+e^{x\bA_l}e^{4it\tzs_3\bA_l^2}
\bP_le^{x\bA_l}\right]^{-1},
\end{align*}
provided the inverse matrices exist. Here
\begin{align*}
\bP_r&=\int_0^\infty dz\,e^{-z\bA_r}\bB_r\bC_re^{-z\bA_r}=\begin{pmatrix}
0_{p\times p}&\CQ_{r,\zS}\\ \CN_{r,\Xi}&0_{p\times p}\end{pmatrix},\\
\bP_l&=\int_0^\infty dz\,e^{-z\bA_l}\bB_l\bC_le^{-z\bA_l}=\begin{pmatrix}
0_{p\times p}&\CN_{l,\Xi}\\ \CQ_{l,\zS}&0_{p\times p}\end{pmatrix},
\end{align*}
are the unique solutions of the Sylvester equations
\begin{equation}\label{7.5}
\bA_r\bP_r+\bP_r\bA_r=\bB_r\bC_r,\qquad\bA_l\bP_l+\bP_l\bA_l=\bB_l\bC_l.
\end{equation}
More precisely, given $(x,t)\in\R^2$, the Marchenko integral equations
\eqref{4.10} are uniquely solvable (in an $L^1$-setting) iff the algebraic
equations for $\bW_{r,l}(x;t)$ are uniquely solvable. Consequently,
\begin{subequations}\label{7.6}
\begin{align}
K(x,y,t)&=-\bC_re^{-x\bA_r}
\left[I_{2p}+e^{-x\bA_r}e^{4it\tzs_3\bA_r^2}\bP_re^{-x\bA_r}\right]^{-1}
e^{-y\bA_r}e^{4it\tzs_3\bA_r^2}\bB_r\nonumber\\
&=-\bC_r\left[I_{2p}+e^{-2x\bA_r}e^{4it\tzs_3\bA_r^2}\bP_r\right]^{-1}
e^{-(x+y)\bA_r}e^{4it\tzs_3\bA_r^2}\bB_r\nonumber\\
&=-\bC_r\bPi_r(x;t)^{-1}e^{-(y-x)\bA_r}\bB_r,\label{7.6a}\\
J(x,y,t)&=-\bC_le^{x\bA_l}
\left[I_{2p}+e^{x\bA_l}e^{4it\tzs_3\bA_l^2}\bP_le^{x\bA_l}\right]^{-1}
e^{y\bA_l}e^{4it\tzs_3\bA_l^2}\bB_l\nonumber\\
&=-\bC_l\left[I_{2p}+e^{2x\bA_l}e^{4it\tzs_3\bA_l^2}\bP_l\right]^{-1}
e^{(x+y)\bA_l}e^{4it\tzs_3\bA_l^2}\bB_l\nonumber\\
&=-\bC_l\bPi_l(x;t)^{-1}e^{-(x-y)\bA_l}\bB_l,\label{7.6b}
\end{align}
\end{subequations}
where
\begin{equation}\label{7.7}
\bPi_r(x;t)=e^{2x\bA_r}e^{-4it\tzs_3\bA_r^2}+\bP_r,\qquad
\bPi_l(x;t)=e^{-2x\bA_l}e^{-4it\tzs_3\bA_l^2}+\bP_l.
\end{equation}
Observe that, for $(x,t)\in\R^2$, the matrices $\bPi_{r,l}(x;t)$ are invertible
iff the Marchenko integral equations \eqref{4.10} are uniquely solvable. Using
\eqref{3.5} [with $\bQ=\CQ^2+\CQ_x$ instead of $\CQ$] and \eqref{3.19}, we get
\begin{subequations}\label{7.8}
\begin{align}
-\CQ(x;t)+\int_x^\infty dy\,\CQ(y;t)^2&=2K(x,x;t)
=-2\bC_r\bPi_r(x;t)^{-1}\bB_r,\label{7.8a}\\
\CQ(x;t)+\int_{-\infty}^x dy\,\CQ(y;t)^2&=2J(x,x;t)
=-2\bC_l\bPi_l(x;t)^{-1}\bB_l,\label{7.8b}
\end{align}
\end{subequations}
where $\bPi_r(x;t)$ and $\bPi_l(x;t)$ are defined by \eqref{7.7}. The matrix
NLS solution then equals the $\mp$ block off-diagonal parts of \eqref{7.8},
where we recover the explicit (matrix) NLS solutions obtained before by the
matrix triplet method \cite{ADM,DM2,SIMAI}. Differentiating \eqref{7.8} with
respect to $x$ and using \eqref{1.4} we obtain the following solutions of the
nonlocal integrable equation \eqref{2.2}:
\begin{subequations}\label{7.9}
\begin{align}
\bQ(x;t)&=-\frac{d}{dx}\left(-2\bC_r\bPi_r(x;t)^{-1}\bB_r\right)\nonumber\\
&=-4\bC_r\bPi_r(x;t)^{-1}\bA_re^{2x\bA_r}e^{-4it\tzs_3\bA_r^2}\bPi_r(x;t)^{-1}
\bB_r,\label{7.9a}\\
\bQ(x;t)&=+\frac{d}{dx}\left(-2\bC_l\bPi_l(x;t)^{-1}\bB_l\right)\nonumber\\
&=-4\bC_l\bPi_l(x;t)^{-1}\bA_le^{-2x\bA_l}e^{-4it\tzs_3\bA_r^2}
\bPi_l(x;t)^{-1}\bB_l.\label{7.9b}
\end{align}
\end{subequations}


\section*{Acknowledgments}
The authors have been partially supported by the Regione Autonoma della Sardegna
research project {\it Algorithms and Models for Imaging Science} [AMIS]
(RASSR57257, intervento finanziato con risorse FSC 2014-2020 -- Patto per lo
Sviluppo della Regione Sardegna), and by INdAM-GNFM.


\appendix

\section{Deriving the nonlocal integrable equation}\label{sec:A}

In this appendix we present an alternative proof of \eqref{2.3}. Indeed, let us
start from the matrix NLS equation
$$i\bzs_3\CQ_t+\CQ_{xx}-2\CQ^3=0_{n\times n},$$
where we seek solutions $\CQ$ that anticommute with $\bzs_3$. Putting
$$\bQ=\CQ^2+\CQ_x,$$
so that
\begin{align*}
\bQ_x&=\CQ_x\CQ+\CQ\CQ_x+\CQ_{xx},\\
\bQ_{xx}&=\CQ_{xx}\CQ+2\CQ_x^2+\CQ\CQ_{xx}+\CQ_{xxx},
\end{align*}
we obtain
\begin{align*}
-i\bzs_3\bQ_t&=-(i\bzs_3\CQ_t)\CQ+\CQ(i\bzs_3\CQ_t)-(i\bzs_3\CQ_t)_x\\
&=(\CQ_{xx}-2\CQ^3)\CQ-\CQ(\CQ_{xx}-2\CQ^3)+\CQ_{xxx}-2(\CQ^3)_x\\
&=\CQ_{xx}\CQ-\CQ\CQ_{xx}+\CQ_{xxx}-2\CQ_x\CQ^2-2\CQ\CQ_x\CQ-2\CQ^2\CQ_x\\
&=(\bQ_x-\CQ_x\CQ-\CQ\CQ_x)\CQ-\CQ(\bQ_x-\CQ_x\CQ-\CQ\CQ_x)\\
&+(\bQ_{xx}-\CQ_{xx}\CQ-2\CQ_x^2-\CQ\CQ_{xx})
-2\CQ_x\CQ^2-2\CQ\CQ_x\CQ-2\CQ^2\CQ_x\\
&=\bQ_x\CQ-\CQ\bQ_x+\bQ_{xx}\\
&-3\CQ_x\CQ^2-2\CQ\CQ_x\CQ-\CQ^2\CQ_x-\CQ_{xx}\CQ-2\CQ_x^2-\CQ\CQ_{xx}\\
&=\bQ_x\CQ-\CQ\bQ_x+\bQ_{xx}-\bQ_x\CQ-\CQ\bQ_x-2\CQ_x\CQ^2-2\CQ_x^2\\
&=\bQ_{xx}-2\CQ\bQ_x-2\CQ_x\CQ^2-2\CQ_x^2\\
&=\bQ_{xx}-2\CQ\bQ_x-2\CQ_x\bQ,
\end{align*}
where $\CQ_x=\hf(\bQ-\bzs_3\bQ\bzs_3)$ and hence
$$\CQ(x;t)=-\int_x^\infty dy\,\hf\left(\bQ-\bzs_3\bQ\bzs_3\right)
=\int_{-\infty}^x dy\,\hf\left(\bQ-\bzs_3\bQ\bzs_3\right).$$
Here the integrand depends on $(y,t)\in\R^2$. We have thus reproduced
\eqref{2.3} by a more direct method.

\section{Adjoint symmetry of Marchenko kernels}\label{sec:B}

In this appendix we prove the adjoint symmetry relations \eqref{4.12} for the
Marchenko integral kernels and \eqref{4.13} for the norming constants. Even
though these relations are easily established for scalar Schr\"o\-din\-ger
equations \cite{Fa64,DT79,CD82,CS89} or for matrix Schr\"o\-din\-ger equations
without discrete eigenvalues (where they are immediate from \eqref{3.14}), they
require considerably more effort in the general matrix Schr\"o\-din\-ger case.

Let $G_l(x,k)$ and $G_r(x,k)$ be the $n\times n$ matrix solutions of the dual
matrix Schr\"o\-din\-ger equation
\begin{equation}\label{B.1}
-\bphi^{\prime\prime}(x,k)+\bphi(x,k)\bQ(x)=k^2\bphi(x,k)
\end{equation}
under the asymptotic conditions
\begin{subequations}\label{B.2}
\begin{alignat}{3}
G_l(x,k)&=e^{ik x}\left[I_n+o(1)\right],&\qquad&x\to+\infty,\label{B.2a}\\
G_r(x,k)&=e^{-ik x}\left[I_n+o(1)\right],&\qquad&x\to-\infty.\label{B.2b}
\end{alignat}
\end{subequations}
Then the dual Faddeev functions satisfy the Volterra integral equations
\begin{subequations}\label{B.3}
\begin{align}
e^{-ik x}G_l(x,k)&=I_n+\int_x^\infty dy\,
\frac{e^{2ik(y-x)}-1}{2ik}e^{-ik y}G_l(y,k)\bQ(y),\label{B.3a}\\
e^{ik x}G_r(x,k)&=I_n+\int_{-\infty}^x dy\,
\frac{e^{2ik(x-y)}-1}{2ik}e^{ik y}G_r(y,k)\bQ(y).\label{B.3b}
\end{align}
\end{subequations}
In analogy with Proposition \ref{th:3.1} we prove

\begin{proposition}\label{th:B.1}
The dual Faddeev functions $e^{-ik x}G_l(x,k)$ and $e^{ik x}G_r(x,k)$ are
continuous in $k\in\C^+\cup\R$, are analytic in $k\in\C^+$, and tend to $I_n$ as
$k\to\infty$ from within $\C^+\cup\R$, irrespective of the choice of $x\in\R$.
\end{proposition}

Using the symmetry relation \eqref{1.2} in the matrix Schr\"o\-din\-ger equation
\eqref{1.5} with $k$ replaced by $-k^*$, we see that
$\bzs_3\Psi(x,-k^*)^\dagger\bzs_3$ satisfies the dual matrix Schr\"o\-din\-ger
equation \eqref{B.1} whenever $\Psi(x,k)$ satisfies the original matrix
Schr\"o\-din\-ger equation \eqref{1.5}. Equations \eqref{3.1} and \eqref{B.2}
then imply the following adjoint symmetry relations:
\begin{equation}\label{B.4}
G_{r,l}(x,k)=\bzs_3F_{r,l}(x,-k^*)^\dagger\bzs_3.
\end{equation}
These relations can also be derived directly from the Volterra integral
equations \eqref{3.2} and \eqref{B.3}. With the help of \eqref{3.3},
\eqref{3.10}, and \eqref{B.4} we obtain for $0\neq k\in\R$ the asymptotic
relations
\begin{alignat*}{3}
G_l(x,k)&=e^{ik x}A_r(k)-e^{-ik x}B_r(-k)+o(1),&\qquad&x\to-\infty,\\
G_r(x,k)&=e^{-ik x}A_l(k)-e^{ik x}B_l(-k)+o(1),&\qquad&x\to+\infty.
\end{alignat*}

Putting
\begin{align*}
\bG_r(x,k)&=\begin{pmatrix}G_r^\prime(x,-k)&-G_r(x,-k)\\
-G_r^\prime(x,k)&G_r(x,k)\end{pmatrix}
=(\zs_2\otimes\bzs_3)\bF_r(x,-k)^\dagger(\zs_2\otimes\bzs_3),\\
\bG_l(x,k)&=\begin{pmatrix}G_l^\prime(x,k)&-G_l(x,k)\\
-G_l^\prime(x,-k)&G_l(x,-k)\end{pmatrix}
=(\zs_2\otimes\bzs_3)\bF_l(x,-k)^\dagger(\zs_2\otimes\bzs_3),
\end{align*}
where $\zs_2\otimes\bzs_3=\left[\begin{smallmatrix}0_{n\times n}&-i\bzs_3\\
i\bzs_3&0_{n\times n}\end{smallmatrix}\right]$, we get for $0\neq k\in\R$
\begin{subequations}\label{B.5}
\begin{align}
\bG_r(x,k)&=\begin{pmatrix}A_l(-k)&B_l(k)\\B_l(-k)&A_l(k)
\end{pmatrix}\bG_l(x,k),\label{B.5a}\\
\bG_l(x,k)&=\begin{pmatrix}A_r(k)&B_r(-k)\\B_r(k)&A_r(-k)\end{pmatrix}
\bG_r(x,k),\label{B.5b}
\end{align}
\end{subequations}
where the matrices containing the $A$ and $B$ quantities are each other's
inverses. By analytic continuation we obtain for $k\in\C^+\cup\R$
$$\bG_{r,l}(x,k)=(\zs_2\otimes\bzs_3)\bF_{r,l}(x,-k^*)^\dagger
(\zs_2\otimes\bzs_3).$$

Using \eqref{B.5} and \eqref{3.12} we easily derive the Riemann-Hilbert problem
\begin{equation}\label{B.6}
\begin{pmatrix}G_r(x,-k)\\G_l(x,-k)\end{pmatrix}
=\begin{pmatrix}A_r(k)^{-1}&-R_l(k)\\-R_r(k)&A_l(k)^{-1}\end{pmatrix}
\begin{pmatrix}G_l(x,k)\\G_r(x,k)\end{pmatrix},
\end{equation}
where we are either in the generic case or in the superexceptional case and
there are no spectral singularities.

Assume that the poles $k_s$ of $A_r(k)^{-1}$ are simple and hence \eqref{4.1} is
true. Rewriting the bottom half of \eqref{B.6} we get
\begin{align}
e^{ik x}&G_l(x,-k)=A_{l0}(k)^{-1}e^{ik x}G_r(x,k)
-e^{2ik x}R_r(k)e^{-ik x}G_l(x,k)\nonumber\\
&=A_{l0}(k)e^{ik x}G_r(x,k)+\sum_{s=1}^N\,\tau_{l;s}
\frac{e^{ik x}G_r(x,k)-e^{ik_sx}G_r(x,k_s)}{k-k_s}\nonumber\\
&+i\sum_{s=1}^N\,\frac{e^{2ik_sx}\N_{r;s}}{k-k_s}e^{-ik x}G_l(x,k)
-e^{2ik x}R_r(k)e^{-ik x}G_l(x,k),\label{B.7}
\end{align}
where
\begin{equation}\label{B.8}
\tau_{l;s}G_r(x,k_s)=i\,\N_{r;s}G_l(x,k_s)
\end{equation}
for the ``dual" norming constants $\N_{r;s}$. By the same token,
\begin{align}
e^{-ik x}&G_r(x,-k)=A_r(k)^{-1}e^{-ik x}G_l(x,k)
-e^{-2ik x}R_l(k)e^{ik x}G_r(x,k)\nonumber\\
&=A_{r0}(k)e^{-ik x}G_l(x,k)+\sum_{s=1}^N\,\tau_{r;s}
\frac{e^{-ik x}G_l(x,k)-e^{-ik_sx}G_l(x,k_s)}{k-k_s}
\nonumber\\&+i\sum_{s=1}^N\,\frac{e^{-2ik_sx}\N_{l;s}}{k-k_s}
e^{ik x}G_r(x,k)\!-\!e^{-2ik x}R_l(k)e^{ik x}G_r(x,k),\label{B.9}
\end{align}
where
\begin{equation}\label{B.10}
\tau_{r;s}G_l(x,k_s)=i\,\N_{l;s}G_r(x,k_s)
\end{equation}
for another set of ``dual" norming constants $\N_{l;s}$. For every $x\in\R$ the
identities \eqref{B.7} and \eqref{B.9} are equations in $\CW^{n\times n}$. Using
$\Pi_-$ to project these two equations onto $\CW_-^{n\times n}$ along
$\CW_{+,0}^{n\times n}$, we obtain
\begin{subequations}\label{B.11}
\begin{align}
e^{ik x}G_l(x,-k)&=I_n+i\sum_{s=1}^N\,\frac{e^{2ik_sx}\N_{r;s}}{k-k_s}
e^{-ik x}G_l(x,k)\nonumber\\&-\Pi_-\left[e^{2ik x}R_r(k)
e^{-ik x}G_l(x,k)\right],\label{B.11a}\\
e^{-ik x}G_r(x,-k)&=I_n+i\sum_{s=1}^N\,\frac{e^{-2ik_sx}
\N_{l;s}}{k-k_s}e^{ik x}G_r(x,k)\nonumber\\
&-\Pi_-\left[e^{-2ik x}R_l(k)e^{ik x}G_r(x,k)\right].\label{B.11b}
\end{align}
\end{subequations}
Using the triangular representations
\begin{subequations}\label{B.12}
\begin{align}
e^{-ik x}G_l(x,k)&=I_n+\int_x^\infty dy\,e^{ik(y-x)}\K(y,x),\label{B.12a}\\
e^{ik x}G_r(x,k)&=I_n+\int_{-\infty}^x dy\,e^{ik(x-y)}\J(y,x),\label{B.12b}
\end{align}
\end{subequations}
in \eqref{B.11} and stripping off the Fourier transforms we get the dual
Marchenko equations
\begin{subequations}\label{B.13}
\begin{align}
\K(y,x)+\zo_r(y+x)&+\int_x^\infty dz\,\zo_r(y+z)\K(z,x)=0_{n\times n},
\label{B.13a}\\
\J(y,x)+\zo_l(y+x)&+\int_{-\infty}^x dz\,\zo_l(y+z)\J(z,x)=0_{n\times n},
\label{B.13b}
\end{align}
\end{subequations}
where the dual Marchenko kernels are given by
\begin{align*}
\zo_r(w)&=\hat{R}_r(w)+\sum_{s=1}^N\,e^{ik_sw}\N_{r;s},\\
\zo_l(w)&=\hat{R}_l(w)+\sum_{s=1}^N\,e^{-ik_sw}\N_{l;s}.
\end{align*}

\begin{theorem}\label{th:B.2}
Let the generic case or the superexceptional case be satisfied and assume there
do not exist any spectral singularities. Then the dual Marchenko kernels are
related to the Marchenko kernels as follows:
\begin{equation}\label{B.14}
\zo_r(w)=\bzs_3\zO_r(w)^\dagger\bzs_3,\qquad
\zo_l(w)=\bzs_3\zO_l(w)^\dagger\bzs_3.
\end{equation}
\end{theorem}

\begin{proof}
Taking the adjoint in \eqref{4.10a}, we see that the matrix functions
$\bzs_3\zO_{l,r}(w)^\dagger\bzs_3$ satisfy the dual Marchenko equations
\eqref{B.13}, where
$$\K(y,x)=\bzs_3K(x,y)^\dagger\bzs_3,\qquad\J(y,x)=\bzs_3J(x,y)^\dagger\bzs_3.$$
Since the Marchenko equation \eqref{B.13a} is uniquely solvable in
$L^1(x,+\infty)^{n\times n}$ for large enough $x$ and \eqref{B.13b} is uniquely
solvable in $L^1(-\infty,x)^{n\times n}$ for large enough $-x$, we obtain the
symmetry relations
\begin{equation}\label{B.15}
\zo_r(w)=\bzs_3\zO_r(w)^\dagger\bzs_3,\qquad\zo_l(w)
=\bzs_3\zO_l(w)^\dagger\bzs_3,
\end{equation}
where $w\ge x_0>-\infty$ in the first identity and $w\le x_0<+\infty$ in the
second identity. Since the Fourier transformed reflection coefficients satisfy
$$\hat{R}_r(\za)^\dagger=\bzs_3\hat{R}_r(\za)\bzs_3,\qquad
\hat{R}_l(\za)^\dagger=\bzs_3\hat{R}_l(\za)\bzs_3,$$
we see that $\zO_r(w)-\hat{R}_r(w)$ is an entire analytic matrix function of $w$
that satisfies \eqref{B.15} for $w\ge x_0>-\infty$. Similarly,
$\zO_l(w)-\hat{R}_l(w)$ is an entire analytic matrix function that satisfies
\eqref{B.15} for $w\le x_0<+\infty$. By analytic continuation, \eqref{B.14}
holds true.
\end{proof}

We immediately have the symmetry relations for the norming constants
$$N_{r;s}=\bzs_3\N_{r;s}^\dagger\bzs_3,\qquad
N_{l;s}=\bzs_3\N_{l;s}^\dagger\bzs_3,$$
provided $k_s$ is a simple pole of $A_r(k)^{-1}$.

\begin{theorem}\label{th:B.3}
Let the generic case or the superexceptional case be satisfied and assume there
do not exist any spectral singularities. Then the dual Marchenko kernels satisfy
$$\zo_r(w)=\zO_r(w),\qquad\zo_l(w)=\zO_l(w).$$
\end{theorem}

\begin{proof}
Using \eqref{3.16} and \eqref{B.12} we easily verify that for $0\neq k\in\R$
\begin{align*}
 F_r(x,k)G_r(x,-k)&-F_r(x,-k)G_r(x,k)\nonumber\\
&=F_r(x,k)\left[A_r(k)^{-1}G_l(x,k)-R_l(k)G_r(x,k)\right]\nonumber\\
&-\left[-F_r(x,k)R_l(k)+F_l(x,k)A_l(k)^{-1}\right]G_r(x,k)\nonumber\\
&=F_r(x,k)A_r(k)^{-1}G_l(x,k)-F_l(x,k)A_l(k)^{-1}G_r(x,k),
\end{align*}
which can be meromorphically extended to the upper half-plane. In the same way
we get
\begin{align*}
 F_l(x,k)G_l(x,-k)&-F_l(x,-k)G_l(x,k)\nonumber\\
&=F_l(x,k)\left[-R_r(k)G_l(x,k)+A_l(k)^{-1}G_r(x,k)\right]\nonumber\\
&-\left[F_r(x,k)A_r(k)^{-1}-F_l(x,k)R_r(k)\right]G_l(x,k)\nonumber\\
&=F_l(x,k)A_l(k)^{-1}G_r(x,k)-F_r(x,k)A_r(k)^{-1}G_l(x,k),
\end{align*}
which can be meromorphically extended to the upper half-plane. If $k_s$ is
a simple pole of $A_r(k)^{-1}$ and hence of $A_l(k)^{-1}$, we get by taking
the residues
\begin{equation}\label{B.16}
 F_r(x,k_s)\tau_{r;s}G_l(x,k_s)-F_l(x,k_s)\tau_{l;s}G_r(x,k_s)
\end{equation}
as well as the negative of this expression. Note that either term in the
expression \eqref{B.16} is exponentially decaying as $x\to\pm\infty$. Using
\eqref{4.6}, \eqref{4.8}, \eqref{B.8}, and \eqref{B.10}, we can write the
latter expression in the two equivalent forms
$$iF_l(x,k_s)\left[N_{r;s}-\N_{r;s}\right]G_l(x,k_s)
=iF_r(x,k_s)\left[\N_{l;s}-N_{l;s}\right]G_r(x,k_s).$$
Utilizing the asymptotic behavior of the Jost and dual Jost matrices as
$x\to\pm\infty$, we obtain $e^{2ik_sx}\left[N_{r;s}-\N_{r;s}\right]$ as
$x\to+\infty$ and $e^{-2ik_sx}\left[\N_{l;s}-N_{l;s}\right]$ as $x\to-\infty$.
Therefore,
$$\N_{r;s}=N_{r;s},\qquad\N_{l;s}=N_{l;s}.$$
Consequently, Theorem \ref{th:B.3} is true if the poles of $A_{r,l}(k)^{-1}$ are
simple.
\end{proof}

\begin{corollary}\label{th:B.4}
Let the generic case or the superexceptional case be satisfied and assume there
do not exist any spectral singularities. Then the Marchenko kernels satisfy
$$\zO_r(w)=\bzs_3\zO_r(w)^\dagger\bzs_3,\qquad
\zO_l(w)=\bzs_3\zO_l(w)^\dagger\bzs_3.$$
In other words, the Marchenko kernels are $\bzs_3$-hermitian.
\end{corollary}


\section{Direct substitution}\label{sec:C}

In this appendix we derive the multisoliton solution $\bQ(x;t)$ of the nonlocal
nonlinear evolution equation \eqref{2.3} in the reflectionless case by
substituting the derivatives of the expressions
\begin{subequations}\label{C.1}
\begin{align}
\int_x^\infty dy\,\bQ(y;t)&=2K(x,x;t)=-2\bC_r\bPi_r(x;t)^{-1}\bB_r,
\label{C.1a}\\
\int_{-\infty}^x dy\,\bQ(y;t)&=2J(x,x;t)=-2\bC_l\bPi_l(x;t)^{-1}\bB_l,
\label{C.1b}
\end{align}
\end{subequations}
with respect to $x\in\R$ directly into \eqref{2.3}, where $\bPi_{r,l}(x;t)$ are
defined by \eqref{7.7}.

Recall that $\bzs_3=I_{m_1}\oplus(-I_{m_2})$ and $\tzs_3=I_p\oplus(-I_p)$,
where $p$ is the order of the matrices $A_{r,l}$ in \eqref{C.1} and \eqref{7.3}
is satisfied. Then \eqref{7.3} implies the intertwining relations
\begin{subequations}\label{C.2}
\begin{align}
\tzs_3\bP_{r,l}&=-\bP_{r,l}\tzs_3,\label{C.2a}\\
\tzs_3\bPi_{r,l}(x;t)&=\tPi_{r,l}(x;t)\tzs_3,\label{C.2b}
\end{align}
\end{subequations}
where
$$\tPi_r(x;t)=e^{2x\bA_r}e^{-4it\tzs_3\bA_r^2}-\bP_r,\qquad
\tPi_l(x;t)=e^{-2x\bA_l}e^{-4it\tzs_3\bA_l^2}-\bP_l.$$

\begin{theorem}\label{th:C.1}
Let $(\bA_r,\bB_r,\bC_r)$ be a matrix triplet satisfying \eqref{7.3} for which
the Sylvester equation \eqref{7.5} has a unique solution $\bP_r$. Then the
matrix function $\bQ(x;t)$ defined by
$$\bQ(x;t)=-4\bC_r\bPi_r(x;t)^{-1}\bA_re^{2x\bA_r}e^{-4it\tzs_3\bA_r^2}
\bPi_r(x;t)^{-1}\bB_r$$
satisfies the nonlinear evolution equation \eqref{2.3} in those $(x,t)\in\R^2$
for which the matrix $\bPi_r(x;t)$ is nonsingular.
\end{theorem}

\begin{proof}
Let us introduce the abbreviations
$$\bE=e^{2x\bA_r}e^{-4it\tzs_3\bA_r^2},\qquad\left[\ldots\right]=\bPi(x;t)
=\bE+\bP_r=e^{2x\bA_r}e^{-4it\tzs_3\bA_r^2}+\bP_r,$$
so that
$$\left[\bullet\right]=\tPi(x;t)=\tzs_3\left[\ldots\right]\tzs_3=\bE-\bP_r
=e^{2x\bA_r}e^{-4it\tzs_3\bA_r^2}-\bP_r.$$
Then $\bE_x=2\bA_r\bE=2\bE\bA_r$. In terms of these abbreviations we now define
the product
$$\zG_{,j_1,j_2,\ldots,j_r,}=\left[\ldots\right]^{-1}\bA_r^{j_1}\bE
\left[\ldots\right]^{-1}\bA_r^{j_2}\bE\left[\ldots\right]^{-1}\ldots
\left[\ldots\right]^{-1}\bA_r^{j_r}\bE\left[\ldots\right]^{-1}$$
consisting of $r+1$ factors $\left[\ldots\right]^{-1}$ interrupted by the
consecutive factors $\bA_r^{j_s}\bE$, where $j_s$ is a nonnegative integer
($s=1,2,\ldots,r$). Then \eqref{7.9a} implies that
\begin{equation}\label{C.3}
\bQ(x;t)=-4\bC_r[\ldots]^{-1}\bA_r\bE[\ldots]^{-1}\bB_r=-4\bC_r\zG_{,1,}\bB_r.
\end{equation}
In the subscript string $,j_1,\ldots,j_r,$ we replace a comma by a semicolon if
the factor $[\bullet]^{-1}$ takes the place of $[\ldots]^{-1}$. Using
\eqref{7.3} we get from \eqref{C.3}
\begin{equation}\label{C.4}
\bzs_3\bQ(x;t)\bzs_3=4\bC_r[\bullet]^{-1}\bA_r\bE[\bullet]^{-1}\bB_r
=4\bC_r\zG_{;1;}\bB_r.
\end{equation}

Using \eqref{C.3}, we easily derive the identity
\begin{align*}
\left[\zG_{,j_1,j_2,\ldots,j_r,}\right]_x&=-2\zG_{,1,j_1,\ldots,j_r,}
-2\zG_{,j_1,1,j_2,\ldots,j_r,}-\ldots-2\zG_{,j_1,\ldots,j_r,1,}\nonumber\\
&+2\zG_{,j_1+1,j_2,\ldots,j_r,}+2\zG_{,j_1,j_2+1,j_3,\ldots,j_r,}+\ldots
+2\zG_{,j_1,\ldots,j_{r-1},j_r+1,}
\end{align*}
where we have first differentiated the $r+1$ factors $\left[\ldots\right]^{-1}$
and then the factors $\bA_r^{j_s}\bE$ ($s=1,2,\ldots,r$). A similar
differentiation rule holds if some (or all) of the subscripted commas are
replaced by semicolons.

Using the above $x$-differentiation rule we get
\begin{equation}\label{C.5}
\bQ_x=8\bC_r\left(2\zG_{,1,1,}-\zG_{,2,}\right)\bB_r.
\end{equation}
Using the above $x$-differentiation rule again we get
\begin{equation}\label{C.6}
\bQ_{xx}=16\bC_r\left(-6\zG_{,1,1,1,}+3\zG_{,1,2,}+3\zG_{,2,1,}-\zG_{,3,}
\right)\bB_r.
\end{equation}

Next, using \eqref{C.3} and \eqref{C.4} we compute
\begin{align*}
\CQ_x&=\tfrac{1}{2}(\bQ-\bzs_3\bQ\bzs_3)
=-2\bC_r\zG_{,1,}\bB_r-2\bC_r\zG_{;1;}\bB_r\\
&=\frac{\pt}{\pt x}\bC_r
\left(\left[e^{2x\bA_r}e^{-4it\tzs_3\bA_r^2}+\bP_r\right]^{-1}
+\left[e^{2x\bA_r}e^{-4it\tzs_3\bA_r^2}-\bP_r\right]^{-1}\right)\bB_r,
\end{align*}
implying that
\begin{align}
\CQ(x;t)&=2\bC_r\left[e^{2x\bA_r}e^{-4it\tzs_3\bA_r^2}\pm\bP_r\right]^{-1}\times
\nonumber\\&\times e^{2x\bA_r}e^{-4it\tzs_3\bA_r^2}
\left[e^{2x\bA_r}e^{-4it\tzs_3\bA_r^2}\mp\bP_r\right]^{-1}\bB_r\nonumber\\
&=2\bC_r\zG_{,0;}\bB_r=2\bC_r\zG_{;0,}\bB_r,\label{C.7}
\end{align}
where we have used the identity
\begin{equation}\label{C.8}
[\ldots]^{-1}\bE[\bullet]^{-1}=[\bullet]^{-1}\bE[\ldots]^{-1}
=\frac{1}{2}\left([\ldots]^{-1}+[\bullet]^{-1}\right).
\end{equation}

Using \eqref{7.5} we get
\begin{subequations}\label{C.9}
\begin{align}
\bB_r\bC_r&=\bA_r[\ldots]+[\ldots]\bA_r-2\bA_r\bE,\label{C.9a}\\
\bB_r\bC_r&=\bA_r[\ldots]-[\bullet]\bA_r.\label{C.9b}
\end{align}
\end{subequations}
Equations \eqref{C.3} and \eqref{C.9a} imply the identity
\begin{equation}\label{C.10}
\bQ(x;t)^2=16\bC_r\left(\zG_{,1,2,}+\zG_{,2,1,}-2\zG_{,1,1,1,}\right)\bB_r.
\end{equation}
Equations \eqref{C.9b} and \eqref{7.3} imply the identity
\begin{align}
\bzs_3&\bQ\bzs_3\bQ=-16\bC_r\zG_{;1;}\bB_r\bC_r\zG_{,1,}\bB_r\nonumber\\
&=-16\bC_r\zG_{;1;}\left(\bA_r[\ldots]-[\bullet]\bA_r\right)\zG_{,1,}\bB_r
\nonumber\\
&=-16\bC_r[\bullet]^{-1}\bA_r\bE[\bullet]^{-1}
\left(\bA_r[\ldots]-[\bullet]\bA_r\right)[\ldots]^{-1}\bA_r\bE[\ldots]^{-1}\bB_r
\nonumber\\
&=-16\bC_r[\bullet]^{-1}\left(\bA_r\bE[\bullet]^{-1}\bA_r^2\bE
-\bA_r^2\bE[\ldots]^{-1}\bA_r\bE\right)[\ldots]^{-1}\bB_r\nonumber\\
&=-16\bC_r\zG_{;1;2,}\bB_r+16\bC_r\zG_{;2,1,}\bB_r.\label{C.11}
\end{align}
Equations \eqref{C.5}, \eqref{C.7}, and \eqref{C.9b} imply the identity
\begin{align}
\CQ\bQ_x&=16\bC_r[\ldots]^{-1}\bE[\bullet]^{-1}\bB_r\bC_r\left(2\zG_{,1,1,}
-\zG_{,2,}\right)\bB_r\nonumber\\
&=32\bC_r[\ldots]^{-1}\bE[\bullet]^{-1}(\bA_r[\ldots]-[\bullet]\bA_r)
[\ldots]^{-1}\bA_r\bE[\ldots]^{-1}\bA_r\bE[\ldots]^{-1}\bB_r\nonumber\\
&-16\bC_r[\ldots]^{-1}\bE[\bullet]^{-1}(\bA_r[\ldots]-[\bullet]\bA_r)
[\ldots]^{-1}\bA_r^2\bE[\ldots]^{-1}\bB_r\nonumber\\
&=32\bC_r\zG_{,0;2,1,}\bB_r-32\bC_r\zG_{,1,1,1,}\bB_r-16\bC_r\zG_{,0;3,}\bB_r
+16\bC_r\zG_{,1,2,}\bB_r.\label{C.12}
\end{align}

 Finally, using $\tzs_3[\ldots]\tzs=[\bullet]$ as well as
\begin{align*}
\bE_t&=-4i\tzs_3\bA_r^2\bE,\nonumber\\
([\ldots]^{-1})_t&=-[\ldots]^{-1}(-4i\tzs_3\bA_r^2\bE)[\ldots]^{-1}
=4i\tzs_3[\bullet]^{-1}\bA_r^2\bE[\ldots]^{-1}=4i\tzs_3\zG_{;2,}\nonumber\\
([\bullet]^{-1})_t&=-[\bullet]^{-1}(-4i\tzs_3\bA_r^2\bE)[\bullet]^{-1}
=4i\tzs_3[\ldots]^{-1}\bA_r^2\bE[\bullet]^{-1}=4i\tzs_3\zG_{,2;}
\end{align*}
we compute with the help of \eqref{7.3} and \eqref{C.3}
\begin{align}
i\bzs_3\bQ_t&=-4i\bzs_3\bC_r(\zG_{,1,})_t\bB_r
=-4i\bC_r\tzs_3(\zG_{,1,})_t\bB_r\nonumber\\
&=-4i\bC_r\tzs_3\left([\ldots]^{-1}\bA_r\bE[\ldots]^{-1}\right)_t\bB_r
\nonumber\\
&=16\bC_r\tzs_3\left(\tzs_3\zG_{;2,}\bA_r\bE[\ldots]^{-1}
-[\ldots]^{-1}\tzs_3\bA_r^3\bE[\ldots]^{-1}\right.\nonumber\\
&+\left.[\ldots]^{-1}\bA_r\bE\tzs_3\zG_{;2,}\right)\bB_r\nonumber\\
&=16\bC_r\zG_{;2,1,}\bB_r-16\bC_r\zG_{;3,}\bB_r+16\bC_r\zG_{;1;2,}\bB_r.
\label{C.13}
\end{align}

To prove \eqref{2.3} we now employ \eqref{C.13}, \eqref{C.6}, \eqref{C.10},
\eqref{C.11}, and \eqref{C.12} to compute
\begin{align}
i\bzs_3\bQ_t&+\bQ_{xx}-\bQ^2+\bzs_3\bQ\bzs_3\bQ-2\CQ\bQ_x\nonumber\\
&=16\bC_r\zG_{;2,1,}\bB_r-16\bC_r\zG_{;3,}\bB_r
+16\bC_r\zG_{;1;2,}\bB_r\nonumber\\
&-96\bC_r\zG_{,1,1,1,}\bB_r+48\bC_r\zG_{,1,2,}\bB_r
+48\bC_r\zG_{,2,1,}\bB_r-16\bC_r\zG_{,3,}\bB_r\nonumber\\
&-16\bC_r\zG_{,1,2,}\bB_r-16\bC_r\zG_{,2,1,}\bB_r
+32\bC_r\zG_{,1,1,1,}\bB_r\nonumber\\
&-16\bC_r\zG_{;1;2,}\bB_r+16\bC_r\zG_{;2,1,}\bB_r\nonumber\\
&-64\bC_r\zG_{,0;2,1,}\bB_r+64\bC_r\zG_{,1,1,1,}\bB_r
+32\bC_r\zG_{,0;3,}\bB_r-32\bC_r\zG_{,1,2,}\bB_r\nonumber\\
&=32\bC_r\zG_{;2,1,}\bB_r-16\bC_r\zG_{;3,}\bB_r
+32\bC_r\zG_{,2,1,}\bB_r-16\bC_r\zG_{,3,}\bB_r\nonumber\\
&-64\bC_r\zG_{,0;2,1,}\bB_r+32\bC_r\zG_{,0;3,}\bB_r.\label{C.14}
\end{align}

We now regroup the terms in the last member of \eqref{C.14} as follows:
\begin{align}
i\bzs_3\bQ_t&+\bQ_{xx}-\bQ^2+\bzs_3\bQ\bzs_3\bQ-2\CQ\bQ_x\nonumber\\
&=32\bC_r\zG_{;2,1,}\bB_r+32\bC_r\zG_{,2,1,}\bB_r
-64\bC_r\zG_{,0;2,1,}\bB_r\nonumber\\
&-16\bC_r\zG_{;3,}\bB_r-16\bC_r\zG_{,3,}\bB_r
+32\bC_r\zG_{,0;3,}\bB_r.\label{C.15}
\end{align}
Using \eqref{C.8} to derive the identities
$$\zG_{;2,1,}+\zG_{,2,1,}-2\zG_{,0;2,1,}=0_{2p\times2p},\qquad
\zG_{;3,}+\zG_{,3,}-2\zG_{,0;3,}=0_{2p\times2p},$$
it appears that the right-hand side of \eqref{C.15} equals the zero matrix. Thus
the expression for $\bQ(x;t)$ in the statement of Theorem \ref{th:C.1} satisfies
\eqref{2.3}, as claimed.
\end{proof}

By applying Theorem \ref{th:C.1} to the matrix triplet $(-\bA_l,\bB_l,-\bC_l)$
we prove

\begin{theorem}\label{th:C.2}
Let $(\bA_l,\bB_l,\bC_l)$ be a matrix triplet satisfying \eqref{7.3} for which
the Sylvester equation \eqref{7.5} has a unique solution $\bP_l$. Then the
matrix function $\bQ(x;t)$ defined by
$$\bQ(x;t)=-4\bC_l\bPi_l(x;t)^{-1}\bA_le^{-2x\bA_l}e^{-4it\tzs_3\bA_r^2}
\bPi_l(x;t)^{-1}\bB_l$$
satisfies the nonlinear evolution equation \eqref{2.3} in those $(x,t)\in\R^2$
for which the matrix $\bPi_l(x;t)$ is nonsingular.
\end{theorem}


\begin{thebibliography}{WW}
\bibitem{Abl} M.J. Ablowitz, {\it Nonlinear Dispersive Waves. Asymptotic
Analysis and Solitons}, Cambridge Texts in Applied Mathematics {\bf 47},
Cambridge University Press, Cambridge, 2011.
\bibitem{AC} M.J. Ablowitz and P.A. Clarkson, {\it Solitons, Nonlinear Evolution
Equations and Inverse Scattering}, Cambridge University Press, Cambridge, 1991.
\bibitem{AKNS} M.J. Ablowitz, D.J. Kaup, A.C. Newell, and H. Segur, {\it The
inverse scattering transform -- Fourier analysis for nonlinear problems},
Stud. Appl. Math. {\bf 53}, 249--315 (1974).
\bibitem{APT} M.J. Ablowitz, B. Prinari, and A.D. Trubatch, {\it Discrete and
Continuous Nonlinear Schr\"o\-din\-ger Systems}, Cambridge University Press,
Cambridge, 2003.
\bibitem{ASg} M.J. Ablowitz and H. Segur, {\it Solitons and the Inverse
Scattering Transform}, SIAM, Philadelphia, 1981.
\bibitem{AdC} H. Aden and B. Carl, {\it On realizations of solutions of the KdV
equation by determinants on operator ideals}, J. Math. Phys. 37, 1833--1857
(1996).
\bibitem{AM} Z.S. Agranovich and V.A. Marchenko, {\it The Inverse Problem of
Scattering Theory}, Gordon and Breach, New York, 1963.
\bibitem{ADM} T. Aktosun, F. Demontis, and C. van der Mee, {\it Exact solutions
to the focusing nonlinear Schr\"o\-din\-ger equation}, Inverse Problems
{\bf 23}, 2171--2195 (2007).%
\bibitem{ADM2} T. Aktosun, F. Demontis, and C. van der Mee, {\it Exact
solutions to the sine-Gordon equation}, J. Math. Phys. {\bf 51} (2010), 123521,
27 pp.
\bibitem{AKV01} T. Aktosun, M. Klaus, and C. van der Mee, {\it Small-energy
asymptotics of the scattering matrix for the self-adjoint matrix
Schr\"o\-din\-ger operator on the line}, J. Math. Phys. {\bf 42}, 4627--4652
(2001).
\bibitem{AV1} T. Aktosun and C. van der Mee, {\it Explicit solutions to the
Korteweg-de Vries equation on the half-line}, Inverse Problems 22, 2165--2174
(2006).
\bibitem{AW18} T. Aktosun and R. Weder, {\it Inverse scattering on the half
line for the matrix Schr\"o\-din\-ger equation}, Journal of Mathematical
Physics, Analysis, Geometry {\bf 14}, 237--269 (2018).
\bibitem{AW20} T. Aktosun and R. Weder, {\it Direct and Inverse Scattering for
the Matrix Schr\"o\-din\-ger Equation}, Applied Mathematical Sciences {\bf 203},
Springer, New York, 2020.%
\bibitem{Berk} G. Berkolaiko, {\it An elementary introduction to quantum
graphs}. In: {\it Geometric and Computational Spectral Theory}, Contemporary
Mathematics {\bf 700}, Amer. Math. Soc., Providence RI, 2017, pp. 41--72; also:
arXiv:1603.07356v2 (2016).
\bibitem{BCFK} G. Berkolaiko, R. Carlson, S.A. Fulling, and P. Kuchment (eds.),
{\it Quantum Graphs and their Applications}, Contemporary Mathematics
{\bf 415}, Amer. Math. Soc., Providence RI, 2006.
\bibitem{BkKm} G. Berkolaiko and P. Kuchment, {\it Introduction to Quantum
Graphs}, Mathematical Surveys and Monographs {\bf 186}, Amer. Math. Soc.,
Providence RI, 2013.
\bibitem{BkLiu} G. Berkolaiko and W. Liu, {\it Simplicity of eigenvalues and
non-vanishing of eigenfunctions of a quantum graph}, J. Math. Anal. Appl.
{\bf 445}, 803--818 (2017).
\bibitem{BoKu} J. Boman and P. Kurasov, {\it Symmetries of quantum graphs and
the inverse scattering problem}, Adv. Appl. Math. {\bf 35}, 58--70 (2005).
\bibitem{CD82} F. Calogero and A. Degasperis, {\it Spectral Transform and
Solitons},~I, Studies in Mathematics and its Applications {\bf 13}, Elsevier,
Amsterdam and New York, 1982.
\bibitem{CS89} K. Chadan and P. Sabatier, {\it Inverse Problems in Quantum
Scattering Theory}, 2nd ed., Springer, New York, 1989.
\bibitem{CTNP} M. Chen, M.A. Tsankov, J.M. Nash, and C.E. Patton,
{\it Backward-volume-water microwave-envelope solitons in yttrium iron garnet
films}, Phys. Rev. B {\bf 49}, 12773--12790 (1994).
\bibitem{DT79} P. Deift and E. Trubowitz, {\it Inverse scattering on the
line}, Commun. Pure Appl. Math. {\bf 32}, 121--251 (1979).
\bibitem{D2} F. Demontis, {\it Exact solutions to the modified
Korteweg-de Vries equation}, Theor. Math. Phys. {\bf 168}, 886--897 (2011).
\bibitem{DM1} F. Demontis and C. van der Mee, {\it Marchenko equations and
norming constants of the matrix Zakharov-Shabat system}, Operators and Matrices
{\bf 2}, 79--113 (2008).
\bibitem{DM2} F. Demontis and C. van der Mee, {\it Explicit solutions of the
cubic matrix nonlinear Schr\"o\-din\-ger equation}, Inverse Problems {\bf 24},
02520 (2008), 16 pp.
\bibitem{EH} W. Eckhaus and A. van Harten, {\it The Inverse Scattering
Transformation and the Theory of Solitons}, North-Holland, Amsterdam, 1981.
\bibitem{EKKST} P. Exner, J.P. Keating, P. Kuchment, T. Sunada, and A. Teplyaev
(eds.), {\it Analysis on Graphs and its Applications}, Proceedings of Symposia
in Pure Mathematics {\bf 77}, Amer. Math. Soc., Providence RI, 2008.
\bibitem{Fa64} L.D. Faddeev, {\it Properties of the} $S$-{\it matrix of the
one-di\-men\-sional Schr\"o\-din\-ger equation}, Amer. Math. Soc. Transl.
Series {\bf 2}, {\bf 65}, 139--166 (1964).
\bibitem{FT} L.D. Faddeev and L.A. Takhtajan, {\it Hamiltonian Methods in the
Theory of Solitons}, Classics in Mathematics, Springer, New York, 1987.
\bibitem{GRS} I.M. Gelfand, D.A. Raikov, and G.E. Shilov, {\it Commutative
Normed Rings}, Chelsea Publ., New York, 1964.
\bibitem{Ger} N.I. Gerasimenko, {\it The inverse scattering problem on
a noncompact graph}, Theor. Math. Phys. {\bf 75}, 460--470 (1988).
\bibitem{GerPav} I. Gerasimenko and B.S. Pavlov, {\it A scattering problem on
a noncompact graphs}, Theor. Math. Phys. {\bf 74}, 230--240 (1988).
\bibitem{GtSm} B. Gutkin and U. Smilansky, {\it Can one hear the shape of
a graph}?, J. Phys. A: Math. Gen. {\bf 34}, 6061--6068 (2001).
\bibitem{Hm02} M.S. Harmer, {\it Inverse scattering for the matrix
Schr\"o\-din\-ger operator and Schr\"o\-din\-ger operator on graphs with
general self-adjoint boundary conditions}, ANZIAM J. {\bf 44}, 161--168 (2002).
\bibitem{Hm04} M.S. Harmer, {\it The matrix Schr\"o\-din\-ger operator and
Schr\"o\-din\-ger operator on graphs}, Ph.D. thesis, University of Auckland,
New Zealand, 2004.
\bibitem{Hm05} M.S. Harmer, {\it Inverse scattering on matrices with boundary
conditions}, J. Phys. A: Math. Gen. {\bf 38}, 4875--4885 (2005).
\bibitem{Has} A. Hasegawa, {\it Optical Solitons in Fibers}, Springer Series
in Photonics {\bf 9}, Springer, New York, 2002.
\bibitem{HT} A. Hasegawa and F. Tappert, {\it Transmission of stationary
nonlinear optical pulses in dispersive dielectric fibers}. I.~{\it Anomalous
dispersion}, Appl. Phys. Lett. {\bf 23}(3), 142--144 (1973); II.~{\it Normal
dispersion}, Appl. Phys. Lett. {\bf 23}(4), 171--172 (1973).
\bibitem{HJ} R.A. Horn and C.J. Johnson, {\it Topics in Matrix Analysis},
Cambridge University Press, Cambridge, 1994.
\bibitem{KFCG} P.G. Kevrekidis, D.J. Frantzeskakis, R. Carretero-Gonz\'alez,
{\it Emergent Non-linear Phenomena in Bose-Einstein Condensates}, Springer,
Berlin, 2008.
\bibitem{K88} M. Klaus, {\it Low-energy behaviour of the scattering matrix for
the Schr\"o\-din\-ger equation on the line}, Inverse Problems {\bf 4}, 505--512
(1988).
\bibitem{KS99} V. Kostrykin and R. Schrader, {\it Kirchhoff's rule for quantum
wires}, J. Phys. A: Math. Gen. {\bf 32}, 595--630 (1999).
\bibitem{KS00} V. Kostrykin and R. Schrader, {\it Kirchhoff's rule for quantum
wires}. II.~{\it The inverse problem with possible applications to quantum
computers}, Fortschr. Phys. {\bf 48}, 703--716 (2000).
\bibitem{Km04} P. Kuchment, {\it Quantum graphs}. I.~{\it Some basic
structures}, Waves Random Media {\bf 14}, S107--S128 (2004).
\bibitem{Km05} P. Kuchment, {\it Quantum graphs}. II.~{\it Some spectral
properties of quantum and combinatorial graphs}, J. Phys. A: Math. Gen.
{\bf 38}, 4887--4900 (2005).
\bibitem{KuNw05} P. Kurasov and M. Nowaczyk, {\it Geometric properties of
quantum graphs and vertex scattering matrices}, Opusc. Math. {\bf 30}, 295--309
(2010).
\bibitem{KuNw10} P. Kurasov and M. Nowaczyk, {\it Inverse spectral problem for
quantum graphs}, J. Phys. A: Math. Gen. {\bf 38}, 4901--4915 (2005).
\bibitem{KuSt} P. Kurasov and F. Stenberg, {\it On the inverse scattering
problem on branching graphs}, J. Phys. A: Math. Gen. {\bf 35}, 101--121 (2002).
\bibitem{Lx} P. Lax, {\it Integrals of nonlinear equations of evolution and
solitary waves}, Commun. Pure Appl. Math. {\bf 21}, 467--490 (1968).
\bibitem{PS1} C.J. Pethick and H. Smith, {\it Bose-Einstein Condensation in
Dilute Gases}, $2^{\text{\tiny nd}}$ ed., Cambridge University Press, Cambridge,
2008.
\bibitem{PS2} L.P. Pitaevskii and S. Stringari, {\it Bose-Einstein Condensation
and Superconductivity}, Oxford University Press, Oxford, 2016.
\bibitem{Sb1} C. Schiebold, {\it An operator theoretic approach to the Toda
lattice equation}, Physica D {\bf 122}, 37--61 (1998).
\bibitem{Sb2} C. Schiebold, {\it Solutions of the sine-Gordon equation coming
in clusters}, Revista Matem\'atica Complutense {\bf 15}, 265--325 (2002).
\bibitem{Shaw} J.K. Shaw, {\it Mathematical Principles of Optical Fiber
Communications}, CBMS-NSF Regional Conference Series in Applied Mathematics
{\bf 76}, SIAM, Philadelphia, 2004.
\bibitem{SIMAI} C. van der Mee, {\it Nonlinear Evolution Models of Integrable
Type}, SIMAI e-Lecture Notes {\bf 11}, SIMAI, Torino, 2013.
\bibitem{Wd} M. Wadati and T. Kamijo, {\it On the extension of inverse
scattering method}, Prog. Theor. Phys. {\bf 52}, 397--414 (1974).
\bibitem{Z} V.E. Zakharov, {\it Hamilton formalism for hydrodynamic plasma
models}, Sov. Phys. JETP {\bf 33}, 927--932 (1971).
\bibitem{ZP} V.E. Zakharov and A.F. Popkov, {\it Contribution to the nonlinear
theory of magnetostatic spin waves}, Sov. Phys. JETP {\bf 57}, 350--355 (1983).
\bibitem{ZS72} V.E. Zakharov and A.B. Shabat, {\it Exact theory of
two-dimensional self-focusing and one-dimensional self-modulation of waves in
nonlinear media}, Sov. Phys. JETP {\bf 34}, 62--69 (1972).
\end{thebibliography}
\end{document}